\documentclass[prx,preprint,superscriptaddress]{revtex4-2}
\usepackage{graphicx}
\usepackage{amssymb}
\usepackage{amsfonts}
\usepackage{amsmath}
\usepackage{amsthm}
\usepackage{cancel}
\usepackage{tensor}
\usepackage{lmodern}
\usepackage{mathrsfs}
\usepackage{mathtools,slashed,isomath}
\usepackage{mathdots}
\usepackage{collref}
\usepackage[colorlinks=true, citecolor=blue, urlcolor=blue]{hyperref}
\usepackage{tikz}
\usepackage{bm}
\usepackage[markup=underlined]{changes}
\makeatletter
\AddToHook{cmd/added/before}{\def\Changes@AuthorColor{blue}}
\AddToHook{cmd/deleted/before}{\def\Changes@AuthorColor{brown}}
\AddToHook{cmd/replaced/before}{\def\Changes@AuthorColor{purple}}
\makeatother
\usepackage{todonotes}
\setcommentmarkup{\todo[color={pink},size=\scriptsize]{#3: #1}}

\usetikzlibrary{shapes.geometric}
\newtheorem{theorem}{Theorem}
\newtheorem{corollary}{Corollary}[theorem]

\begin{document}
\title{Polynomial complexity of open quantum system problems}
\author{Chong Chen}
\affiliation{Department of Physics, The Chinese University of Hong Kong, Shatin, New Territories, Hong Kong, China}%
\affiliation{The State Key Laboratory of  Quantum Information Technologies and Materials, The Chinese University of Hong Kong, Shatin, New Territories, Hong Kong, China}
\affiliation{New Cornerstone Science Laboratory, The Chinese University of Hong Kong, Shatin, New Territories, Hong Kong, China}%

\author{Ren-Bao Liu}
\email{rbliu@cuhk.edu.hk}
\affiliation{Department of Physics, The Chinese University of Hong Kong, Shatin, New Territories, Hong Kong, China}%
\affiliation{The State Key Laboratory of  Quantum Information Technologies and Materials, The Chinese University of Hong Kong, Shatin, New Territories, Hong Kong, China}%
\affiliation{New Cornerstone Science Laboratory, The Chinese University of Hong Kong, Shatin, New Territories, Hong Kong, China}%
%\affiliation{The Hong Kong Institute of Quantum Information Science and Technology, The Chinese University of Hong Kong, Shatin, New Territories, Hong Kong, China}%
\affiliation{Centre for Quantum Coherence, The Chinese University of Hong Kong, Shatin, New Territories, Hong Kong, China}%

\begin{abstract}
Open quantum systems (OQS's) are ubiquitous in non-equilibrium quantum dynamics and in quantum science and technology. Solving the dynamics of an OQS in a quantum many-body bath has been considered a computationally hard problem because of the dimensionality curse. Here, considering that full knowledge of the bath dynamics is unnecessary for describing the reduced dynamics of an OQS, we prove a polynomial complexity theorem, that is, the number of independent equations required to fully describe the dynamics of an OQS increases at most linearly with the evolution time and polynomially with the bath size. Therefore, efficient computational algorithms exist for solving the dynamics of a small-sized OQS (such as a qubit or an atom). We further prove that, when the dynamics of an OQS and the bath is represented by a tensor network, a tensor contraction procedure can be specified such that the bond dimension (i.e., the range of tensor indices contracted in each step) increases only linearly (rather than exponentially) with the evolution time, providing explicitly efficient algorithms for a wide range of OQS's. We demonstrate the theorems and the tensor-network algorithm by solving two widely encountered OQS problems, namely, a spin in a Gaussian bath (the spin-boson model) and a central spin coupled to many environmental spins (the Gaudin model). This work provides approaches to understanding dynamics of OQS's, learning the environments via quantum sensors, and optimizing quantum information processing in noisy environments.
  
\end{abstract}

\maketitle

\section{Introduction} 

Quantum systems inevitably couple with their environments or baths, becoming open quantum systems (OQS's)~\cite{Breuer2007, Weiss2012}. The OQS problem is of fundamental interest for understanding quantum decoherence~\cite{Schlosshauer2019, Prokof2000}. It is ubiquitous in non-equilibrium quantum dynamics, including quantum transport~\cite{ Gurvitz1996, Li2005, Jin2008} and quantum optical processes~\cite{Scully1997}. It has become a critical problem for the development of modern quantum technologies~\cite{Giovannetti2006, Gisin2007, Bharti2022, Acin2018} such as quantum metrology, quantum communication, and quantum computing, where the coupling of quantum systems to their environments is a main obstacle to maintaining quantum superpositions and scaling up the systems. In quantum sensing, on the other hand, an OQS can act as a quantum probe to its environment~\cite{Bylander2011,Wang2019, Meinel2022, Wu2024}. To design quantum systems and to optimize quantum controls of such systems~\cite{Du2009, Lange2010, Yang2011, Suter2016} for mitigating the decoherence effects and utilizing the quantum advantages it requires good understanding and often the computation of the dynamics of OQS's.

The solution of the dynamics of an OQS in general is believed to be a computationally hard problem (time consumption and/or memory demand increasing at least exponentially with the evolution time and the bath size), since it is assumed to be equivalent to solving the dynamics of the quantum many-body system including the bath. In cases where environments can be taken as white noise, their effects are captured by the Lindblad master equation, known as Markovian dynamics~\cite{Gardiner2004}. However, in general, memory (i.e., non-Markovian) effects are important~\cite{Breuer2016, deVega2017}, especially when the quantum system is strongly coupled to its environment or the environment itself has a structure. Notable examples include a localized electron spin coupled with nuclear spins~\cite{Madsen2011,Mi2017}, an emitter in a photonic crystal~\cite{Thompson2013}, and a multi-qubit system with long-range correlated noises~\cite{Aharonov2005}. The memory effects mean that extra many-body correlations can build up in the baths via coupling to the central system. For example, even a group of non-interacting spins can accumulate many-body correlations when they are commonly coupled to a central spin (such as in the Gaudin model~\cite{Gaudin1976, Cywifmmode2010, Stanek2013, Stanek2014,He2022}). Therefore, the non-Markovian dynamics of OQS's is highly non-trivial.
 
Previous studies indicate that numerically efficient methods may exist for certain OQS's.  In particular, when system evolution and environment dynamics have very different timescales, the problem can be efficiently solved using perturbative theories such as the time-convolutionless projection operator technique~\cite{Breuer2007} and the cluster expansion methods~\cite{Yang2017}. Beyond the perturbation theories, several exact methods have been developed to study OQS's where the environments exhibit simple structures. The pseudo-mode method~\cite{Garraway1997, Mazzola2009, Tamascelli2018, Pleasance2020,Mascherpa2020, Trivedi2021}, the stochastic Schr\"{o}dinger equation~\cite{Diosi1998,Gaspard1999,Stockburger2002,Piilo2008}, the hierarchical equations of motion~\cite{Tanimura1989, Tanimura2020}, and the hierarchy of pure states~\cite{Suess2014, Hartmann2017} are applicable when the environment is Gaussian and has certain spectral characteristics. Path-integral influence functionals~\cite{Makri1995b, Makri1995, Ye2021, Sonner2021, Ng2023, Thoenniss2023} and time-evolving matrix product operators~\cite{Strathearn2018, Strathearn2020, Gribben2022a} can deal with  Gaussian environments that have special system-environment coupling. In addition, the process tensor networks~\cite{Pollock2018,Jorgensen2019} provide frameworks for solving general OQS problems but the sampling of process tensors can be classically hard~\cite{Aloisio2023}.  On the other hand, advanced numerical algorithms developed for quantum many-body systems, such as the numerical renormalization group method~\cite{Bulla2008}, the density matrix renormalization group method~\cite{Prior2010,Stanek2013}, and the tensor network method~\cite{Schroder2016,Wall2016}, have been applied to investigate OQS's in strongly correlated environments~\cite{Orus2014, Schollwock2011, Orus2019, Paeckel2019}, which have significantly mitigated the difficulties of non-Markovian dynamics~\cite{Prior2010, Jorgensen2019, Ye2021, Sonner2021, Ng2023, Thoenniss2023, Link2024}. Furthermore, it has been established (though only formally) that the reduced dynamics of a central quantum system can be {\it fully} captured by the correlations in the bath {\it before} it is coupled to the central system~\cite{Gasbarri2018, Wang2019}; in particular, the dynamics of a spin coupled to a number of harmonic oscillators is fully determined by the spectral density of the harmonic oscillators and is independent of the details of the couplings (for instance, a bath of one harmonic oscillator and that of $10^3$ harmonic oscillators of the same frequency make no difference in the central spin dynamics as long as the summation of the coupling strengths squared  is equal in the two cases, but the dynamics of the whole spin-bath system would be much more complicated in the latter case). All these results suggest that the large number of degrees of freedom of a bath may not necessarily prevent an efficient solution of the non-Markovian dynamics of an OQS~\cite{Luchnikov2019, Cygorek2022, Cygorek2024}. 
However, what the ultimate complexity of the non-Markovian dynamics is and how efficient algorithms, if they exist, can be constructed systematically are still open questions.

Here, we establish a polynomial complexity theorem that the non-Markovian dynamics of a general OQS can be {\it fully} described by a set of independent equations of extended density matrices (EDMs), and the size of the set grows at most linearly with the evolution time and polynomially with the bath size. This means that the complexity of the non-Markovian dynamics increases at most as a polynomial function of the evolution time and the bath size. Moreover, we propose an explicit tensor network algorithm to solve these EDMs when the correlations of environmental noise are known (from computation or measurement). We demonstrate the theorem and the efficiency of the tensor network algorithm by solving two widely encountered classes of OQS problems, namely a spin in a Gaussian bath (the spin-boson model) and a central spin coupled to many environmental spins (the Gaudin model). The numerical results show that the bond dimension of the EDM tensor increases at most linearly with time evolution, regardless of the coupling strength (in the spin-boson model) and the bath size (in the Gaudin model), verifying the efficiency of the tensor network algorithm.

\section{OQS dynamics in terms of bath correlations}
\label{sec_complete_framework}
For a general OQS, the system-bath coupling can be written in the interaction picture with respect to the bath Hamiltonian ${H}_{\rm B}$ and the system Hamiltonian $H_{\rm S}$ as
\begin{equation}\label{eq:Interaction}
%{\hat{H}_{\rm I}}
{H}(t)=\sum^{d^2-1}_{{\alpha}=1} S_{\alpha}(t) \otimes B_{\alpha}(t),
\end{equation}
where $d$ is the dimension of the system's Hilbert space, $\left\{S_{\alpha}(t)\equiv e^{iH_{\rm S}t}S_{\alpha}e^{-iH_{\rm S}t}\right\}$ is a basis of the system operators (with the trivial identity operator $S_0\equiv{I}_S$), and $B_{\alpha}(t)\equiv e^{i {H}_{\rm B} t} B_{\alpha} e^{-i {H}_{\rm B} t}$ is the bath operator coupled to ${S}_{\alpha}$. In principle, the reduced dynamics of the central system can be obtained by first solving the coupled system-bath dynamics (given by the density operator $\chi(t)$) and then carrying out the partial trace of the bath's degrees of freedom~\cite{Breuer2007}. By discretizing the time as $0,\epsilon,2\epsilon,\ldots,T$ with a small time step $\epsilon\equiv T/N$,  the evolution of the total system reads
\begin{equation}\label{eq:vonNeumannE}
\chi(T) = e^{\epsilon \mathcal{H}^{-}(T)}\cdots e^{\epsilon \mathcal{H}^{-}(2\epsilon)} e^{\epsilon \mathcal{H}^{-}(\epsilon)} \chi(0).
\end{equation}
Here, we have introduced the superoperators
$\mathcal{A}^{-} \chi \equiv -i (A \chi- \chi A)$ (essentially a commutator) and $\mathcal{A}^{+} \chi \equiv \frac{1}{2}(A \chi+ \chi A)$ (essentially an anti-commutator) for an operator $A$. A small step of evolution can be expanded to the first order of $\epsilon$ as
\begin{equation}
\label{eq:evolutionOperator}
   e^{\epsilon\mathcal{H}^{-}(t)} \approx {\mathcal I}_{\rm S}\otimes {\mathcal I}_{\rm B} + \epsilon \sum_{{\sigma}=\pm}\sum_{{\alpha}=1}^{d^2-1}{\mathcal S}_{{\alpha}}^{{\sigma}}(t)\otimes {\mathcal B}_{{\alpha}}^{\bar{{\sigma}}}(t)\equiv {\mathcal S}^{\phi}(t)\otimes {\mathcal B}_{\phi}(t), 
\end{equation}
where $\bar{\sigma}\equiv -\sigma$ and a summation over $\phi = 0,1,\ldots, 2(d^2-1)$ is implied in the tensor contraction convention with ${\mathcal S}^{2{\alpha}-1}\equiv \epsilon{\mathcal S}_{\alpha}^{+}$, ${\mathcal S}^{2{\alpha}}\equiv \epsilon {\mathcal S}_{\alpha}^-$, ${\mathcal B}_{2{\alpha}-1}\equiv {\mathcal B}_{\alpha}^-$, ${\mathcal B}_{2{\alpha}}\equiv{\mathcal B}_{\alpha}^+$, and the null evolutions ${\mathcal S}^0\equiv {\mathcal I}_{\rm S}$ for the system and ${\mathcal B}_0\equiv {\mathcal I}_{\rm B}$ for the bath. The physical meanings of the superoperators are clear by checking their effects on the density operators~\cite{Cheung2024}. For example, ${\mathcal S_{\alpha}^-}\rho=-i\left[S_{\alpha},\rho\right]$ causes the evolution of the system and ${\rm Tr} [{\mathcal B}_{\alpha}^+\Omega] =\left\langle B_{\alpha}\right\rangle$, so the term ${\mathcal S}_{\alpha}^-\otimes {\mathcal B}_{\alpha}^+\chi$ corresponds to the evolution of the system driven by the $\alpha$-th component of the force from the bath. Similarly, ${\mathcal S}_{\alpha}^+\otimes {\mathcal B}_{\alpha}^-\chi$ means the evolution of the bath under the $\alpha$-th component of the force from the system. A small step of evolution can be presented by a tensor network diagram, as illustrated in Fig.~\ref{fig:TensorRep}a.

\begin{figure}[tbhp]
    \centering
    \includegraphics[width=0.9\columnwidth]{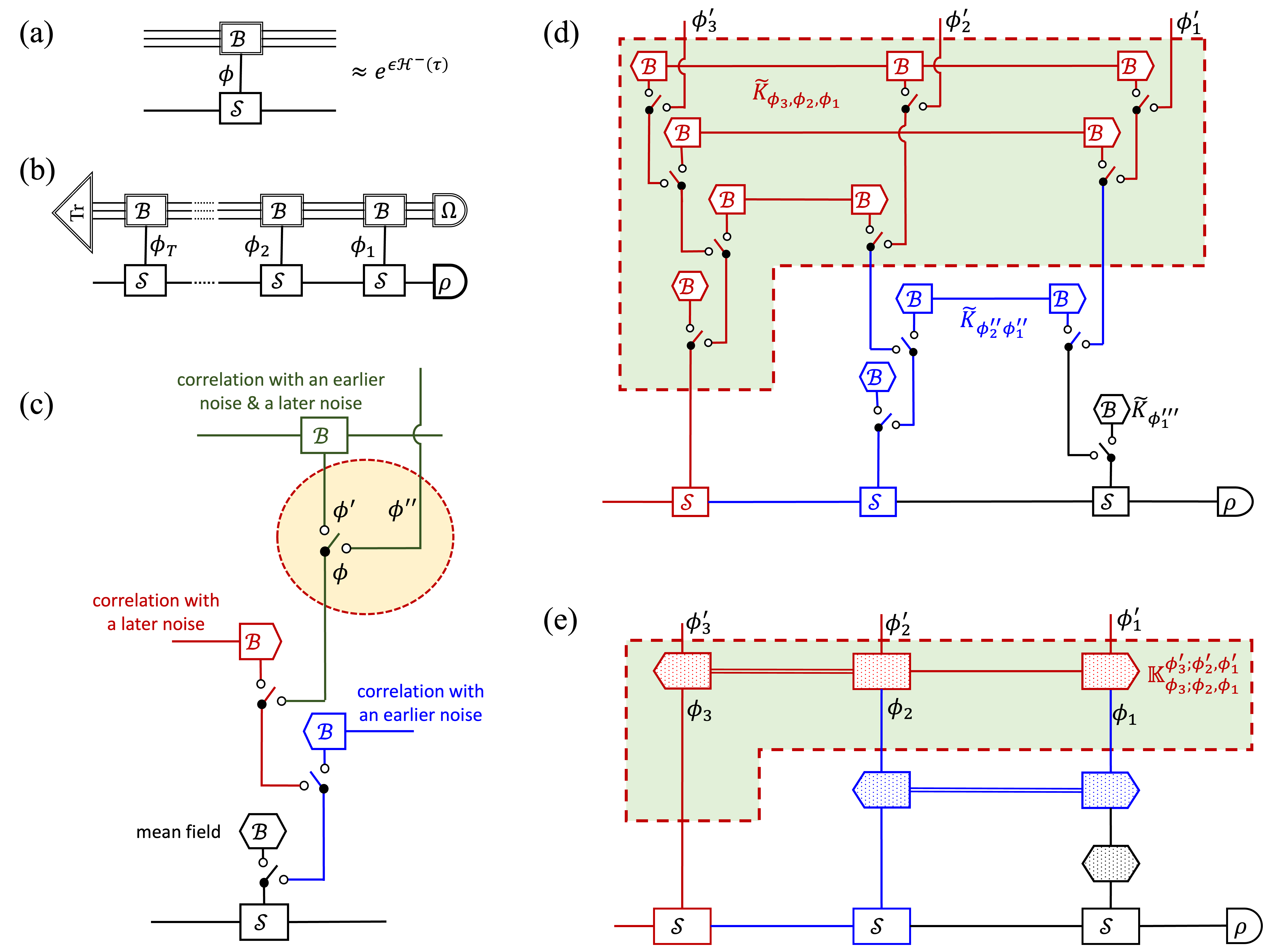}
    \caption{Non-Markovian dynamics of OQS's constructed by correlation and evolution tensors. (a) Tensor network diagram for a small step of  evolution $e^{\epsilon \mathcal{H}^{-}(\tau)}$. (b) Tensor network representation of the reduced dynamics $\rho(T)$ in terms of the bath correlation tensor ${\mathbb C}_{\boldsymbol{\Phi}^{(T)}} \equiv{\rm Tr}_{\rm B}\left[{\mathcal B}_{\phi_T}(T)\ldots {\mathcal B}_{\phi_2}(2){\mathcal B}_{\phi_1}(1)\Omega\right]$  and time-ordered system evolution tensor ${\mathbb S}^{{\boldsymbol{\Phi}^{(T)}}}\equiv {\mathcal S}^{\phi_T}(T)\ldots {\mathcal S}^{\phi_2}(2){\mathcal S}^{\phi_1}(1)$. (c) A noise from the bath at time $t$, given by $\mathcal{B}_{\phi_t}(t)$, can affect the non-Markovian dynamics of the system via different types of irreducible correlation functions, i.e., cumulants, which are selected in the tensor network by a picking tensor ${P}^{\phi',\phi''}_{\phi}$ ( as shown in the circle). (d) The cumulant expansion of the correlation tensor shown by an example of evolution by three small steps. Different cumulants (modified to kernel tensors $\tilde{K}$) are interconnected through picking tensors. The upward open arms with indices $\phi'_1$, $\phi'_2$ and $\phi'_3$ indicate the possibilities that the noises from the bath at time $t=1,2,3$ take effect in correlations with forces at later times. (e) The tensor network representation of the combined kernel tensor that contains all cumulants ended at a certain time $t$, with ${\mathbb K}_{\phi_3;\phi_2,\phi_1}^{\phi'_3;\phi'_2,\phi'_1}$ for $t=3$ shown as the example, corresponding to the shadowed parts in (d) and (e). }
    \label{fig:TensorRep}
\end{figure}

For the sake of simplicity, we consider an initial product state $\chi(0)=\rho(0)\otimes\Omega$ with $\Omega$ and $\rho(0)$ being the initial bath and system states, respectively. The state of the system can be written as~\cite{Gasbarri2018, Wang2019}
\begin{align}
\rho(T)\equiv \text{Tr}_{\rm B}[\chi(T)]={\mathbb C}_{\boldsymbol{\Phi}^{(T)}} {\mathbb S}^{\boldsymbol{\Phi}^{(T)}} \rho(0),
\label{eq_CSrho}
\end{align}
with ${\boldsymbol \Phi}^{(T)}\equiv \left\{\phi_T,\ldots,\phi_2,\phi_1\right\}$, where the time-ordered tensor 
${\mathbb S}^{{\boldsymbol{\Phi}^{(T)}}}\equiv {\mathcal S}^{\phi_t}(T)\cdots {\mathcal S}^{\phi_2}(2){\mathcal S}^{\phi_1}(1)$
denotes the evolution of the system driven by the corresponding correlation in the bath denoted by the correlation tensor
\begin{align}
{\mathbb C}_{\boldsymbol{\Phi}^{(T)}} \equiv{\rm Tr}_{\rm B}\left[{\mathcal B}_{\phi_T}(T)\cdots {\mathcal B}_{\phi_2}(2){\mathcal B}_{\phi_1}(1)\Omega\right].
\label{eq_def_corr}
\end{align}
Hereafter we set the time step $\epsilon$ as unity so that $t$ stands for $\epsilon t$ and $T= N$. The corresponding  tensor-network  representation is shown in Fig.~\ref{fig:TensorRep}b. Note that the result in Eq.~\eqref{eq_CSrho} can be generalized to a correlated initial state $\chi(0)=\sum_{\alpha} \rho_{\alpha}(0)\otimes  \Omega_{\alpha}$ with $\Omega_{\alpha}$ being the $\alpha$-th state of the bath and $\rho_{\alpha}$ a system operator~\cite{Paz-Silva2019}. An important feature of Eq.~\eqref{eq_CSrho} is that the reduced dynamics of the OQS is fully determined by the correlation functions of the bath without including the system-bath interaction as can be seen from
Eq.~\eqref{eq_def_corr}~\cite{Gasbarri2018, Wang2019}. Thus, given the knowledge of the correlation functions of the bath, the dynamics of the system can be calculated with the complexity independent of the size of the bath. However, solving ${\mathbb C}_{\boldsymbol{\Phi}^{(T)}}$ is still a many-body problem as it captures the dynamics of the bath~\cite{White2022,Aloisio2023}. In fact, the number of correlation functions, $(2 d^2-1)^T$,  increases exponentially with the evolution time. 
Therefore, the direct evaluation of the tensor network in Fig.~\ref{fig:TensorRep}b is computationally hard.

Usually, the correlation tensor defined in Eq.~\eqref{eq_def_corr} has a high degree of redundancy. For example, very often, the correlation functions can be reduced to a few irreducible correlation functions, i.e., cumulants and the cumulants have finite memory time. The cumulants $\tilde{C}_{{\boldsymbol\Theta}_i}$ are recursively defined as 
    \begin{equation}\label{eq:cumulantEx}
        {\mathbb C}_{\boldsymbol{\Phi}^{(t)}}\equiv \sum_{{\boldsymbol\Theta}\subseteq \boldsymbol{\Phi}^{(t)}}\delta^{\mathbf 0}_{\bar{\boldsymbol\Theta}}\sum_{\left\{{\boldsymbol\Theta}_i\right\}} \tilde{C}_{{\boldsymbol\Theta}_1} \tilde{C}_{{\boldsymbol\Theta}_2}\cdots \tilde{C}_{{\boldsymbol\Theta}_n},
    \end{equation}
where $\left\{{\boldsymbol\Theta}_i\right\}$ are non-overlapping non-empty subsets of indices selected from ${\boldsymbol \Phi}^{(t)}$ with $\bigcup_{i=1}^n {\boldsymbol\Theta}_i ={\boldsymbol\Theta}\subseteq{\boldsymbol\Phi}^{(t)}$, 
%\added
{$\bar{\boldsymbol\Theta}\equiv {\boldsymbol\Phi}^{(t)}\setminus{\boldsymbol\Theta}$ is complement of the subset ${\boldsymbol\Theta}$, and the summation runs over all subsets ${\boldsymbol\Theta}$ and all possible partitions of ${\boldsymbol\Theta}$.} For example, the cumulnant expansion up to the third-order correlation tensors reads ${\mathbb C}_{\phi_1}\equiv \tilde{C}_{\phi_1}+\delta^0_{\phi_1}$, ${\mathbb C}_{\phi_2, \phi_1}\equiv \tilde{C}_{\phi_2, \phi_1}  + \tilde{C}_{\phi_2} \tilde{C}_{\phi_1}+\delta^0_{\phi_2}\tilde{C}_{\phi_1}+\delta^0_{\phi_1}\tilde{C}_{\phi_2}+\delta^0_{\phi_2}\delta^0_{\phi_1}$, 
and ${\mathbb C}_{\phi_3,\phi_2,\phi_1}\equiv \tilde{ C}_{\phi_3,\phi_2,\phi_1}+\tilde{C}_{\phi_3,\phi_2} \tilde{ C}_{\phi_1}
%+\tilde{C}_{\phi_3,\phi_1} \tilde{C}_{\phi_2}+ \tilde{ C}_{\phi_3} \tilde{C}_{\phi_2, \phi_1}+\tilde{C}_{\phi_3} \tilde{C}_{\phi_2}\tilde{C}_{\phi_1}+\delta^0_{\phi_3}\tilde{C}_{\phi_2,\phi_1}
+ \cdots + \delta^0_{\phi_3}\delta^0_{\phi_2}\delta^0_{\phi_1}$. Note that the cumulant $\tilde{C}_{\boldsymbol \Theta}=0$ when any index $\phi_{\tau}\in{\boldsymbol\Theta}$ takes zero value ($\phi_{\tau}=0$) since if $\phi_{\tau}\in {\boldsymbol \Theta}$ and $\phi_{\tau}=0$ then the correlation $C_{\boldsymbol \Theta}=C_{{\boldsymbol \Theta}\setminus\{\phi_{\tau}\}}$ (reduced to a lower-order correlation). The noise from the bath at a certain time $\mathcal{B}_{\phi_t}(t)$  can affect the non-Markovian dynamics of the system via different types of cumulants, as a time-local mean field (via the first-order irreducible correlation function) or in correlations with noises at earlier and/or later times (via second- or higher-order irreducible correlations), as illustrated in Fig.~\ref{fig:TensorRep}c.  We introduce the picking tensor $P^{\phi',\phi''}_{\phi}= \delta^{\phi'}_{\phi} \delta^{\phi''}_{0}+\delta^{\phi''}_{\phi}\delta^{\phi'}_{0}-\delta^{0}_{\phi} \delta^{\phi'}_{0}\delta^{\phi''}_{0}$ to select how a noise at a certain time participates in non-Markovian dynamics (diagram shown in the circle in Fig.~\ref{fig:TensorRep}c).  When $\phi=0$, the picking tensor selects the null evolution with $\phi''=\phi'=0$; otherwise, it selects one cumulant by setting either $\phi'$ or $\phi''$ equal to $\phi$ and the other upper index equal to 0.
The cumulants in Eq.~\eqref{eq:cumulantEx} are interconnected to form a tensor network, as illustrated in Fig.~\ref{fig:TensorRep}d by
an example of three-step evolution.  Note that to transfer the addition of cumulants in Eq.~\eqref{eq:cumulantEx} into a tensor product form, we have replaced the cumulants with the irreducible kernel tensors $\tilde{K}_{{\boldsymbol\Theta}}\equiv \delta^{{\mathbf 0}}_{{\boldsymbol\Theta}}+\tilde{C}_{{\boldsymbol\Theta}}$.
One can demonstrate that the contraction of these kernel tensors interconnected by the picking tensors correctly yields the correlation tensor [see Appendix~\ref{Appd-sec:KT-app} for the proof].

 An important feature of the non-Markovian dynamics of an OQS is that the noises from the bath at present and earlier times may take delayed effects in correlations with noises at later times. Thus, the reduced density operator $\rho(t)$ is insufficient to keep track of the dynamics of the OQS. Instead, one needs to bookmark different possibilities of how the noise at a certain time would be correlated with future noises, that is, in the tensor network diagram as illustrated in Fig.~\ref{fig:TensorRep}d, to keep the upward open arms stemming from the picking tensors. We denote a tensor network with such upward open arms as an extended density matrix (EDM) at a given time $t$
 \begin{align}
\rho^{{\boldsymbol\Phi}^{(t)}}\equiv \mathbb{C}_{\bm{\Phi}'^{(t)}} \left[\prod^{t}_{\tau=1} {P}^{\phi_{\tau},\phi'_{\tau}}_{\phi''_{\tau}}\right] \mathbb{S}^{\bm{\Phi}''^{(t)}}\rho(0)\equiv \mathbb{C}^{{\boldsymbol\Phi}^{(t)}}_{\bm{\Phi}'^{(t)}} \mathbb{S}^{\bm{\Phi}'^{(t)}}\rho(0),
 \end{align}
where the upper indices ${{\boldsymbol\Phi}^{(t)}}$ keep track of the history (to account for the memory effect), and $\mathbb{C}^{{\boldsymbol\Phi}^{(t)}}_{\bm{\Phi}'^{(t)}}\equiv \mathbb{C}_{\bm{\Phi}''^{(t)}} \prod^{t}_{\tau=1} {P}^{\phi_{\tau},\phi'_{\tau}}_{\phi''_{\tau}} $ (referred to also as correlation tensor with a little abuse of notation).
 For example, the tensor in Fig.~\ref{fig:TensorRep}d represents the EDM $\rho^{\phi'_3,\phi'_2,\phi'_1}$.   
We can further group all kernel tensors ended at time $t$ and the associated picking tensors, keeping track of the upward open arms, into a combined kernel tensor $\mathbb{K}^{\phi'_{t}; {\boldsymbol \Phi}^{'(t-1)}}_{\phi_{t}; {\boldsymbol \Phi}^{(t-1)}}$ [shadowed diagrams in Figs.~\ref{fig:TensorRep}d and ~\ref{fig:TensorRep}e for illustration of ${\mathbb K}_{\phi_3;\phi_2,\phi_1}^{\phi'_3;\phi'_2,\phi'_1}$].

\begin{figure}
    \centering
    \includegraphics[width=0.9\linewidth]{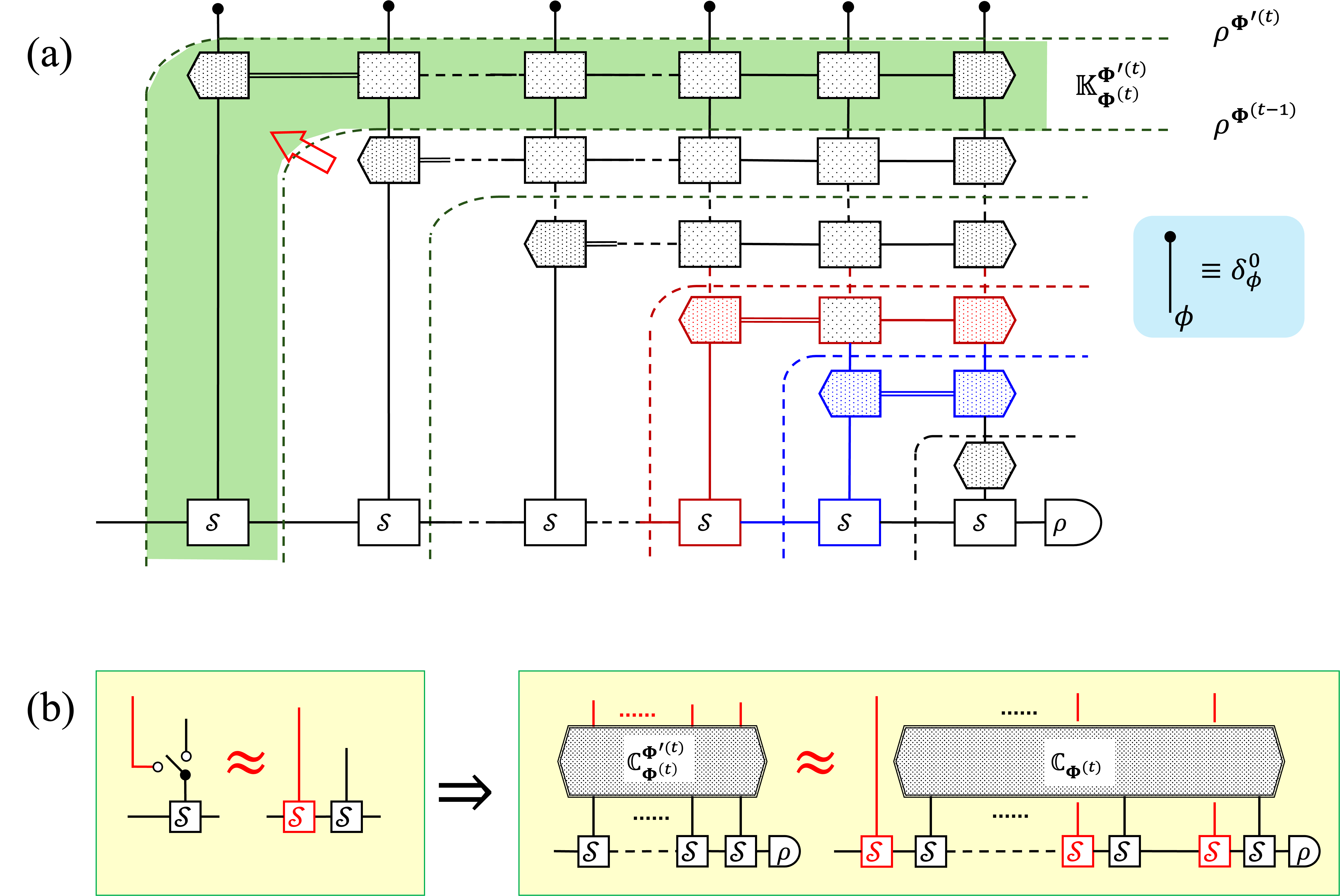}
    \caption{(a) Evolution of extended density matrices constructed by combined kernel tensors.  The arrow indicates the evolution direction. In each step (indicated by a dashed curve), a set of EDMs $\rho^{{\boldsymbol \Phi}^{'(t)}}$ is obtained from $\rho^{{\boldsymbol \Phi}^{(t-1)}}$ by applying the evolution operator $\mathbb{K}^{\phi'_{t};\bm{\Phi}^{'(t-1)}}_{\phi_{t};\bm{\Phi}^{(t-1)}} \mathcal{S}^{\phi_{t}}$  (indicated by the green shadow).   The reduced density matrix $\rho(t)=\delta^{\mathbf 0}_{{\boldsymbol \Phi}'^{(t)}}\rho^{{\boldsymbol \Phi}'^{(t)}}$ is obtained by terminating the upward arms with the closing tensors $\delta^0_{\phi}$ (shown in inset). (b) With the formula ${\mathcal S}^{\phi_{\tau}}{\mathcal S}^{\phi'_{\tau}}=P^{\phi_{\tau},\phi'_{\tau}}_{\phi''_{\tau}}{\mathcal S}^{\phi''_{\tau}}+O\left(\epsilon^2\right)$, the EDM $\rho^{{\boldsymbol \Phi}^{\prime(t)}}$ can be approximated as the density matrix $\rho(t)$ intervened by a sequence of perturbations represented by the superoperators ${\mathcal S}^{\phi'_1}, {\mathcal S}^{\phi'_2}, \ldots, {\mathcal S}^{\phi'_t}$.}
    \label{fig:EDM_evolution}
\end{figure}

The evolution of EDMs is driven by the combined kernel tensors as
\begin{equation}\label{eq:EDM-Evolution}
    \rho^{{\boldsymbol \Phi}^{'(t)}}=\mathbb{K}^{\phi'_{t};\boldsymbol{\Phi}^{'(t-1)}}_{\phi_t; {\boldsymbol \Phi}^{(t-1)}} \mathcal{S}^{\phi_{t}} \rho^{{\boldsymbol \Phi}^{(t-1)}}.
\end{equation}
The tensor network representation of this evolution is shown in Fig.~\ref{fig:TensorRep}e and Fig.~\ref{fig:EDM_evolution}a. Equation~\eqref{eq:EDM-Evolution} serves the role of differential equations for the EDMs, capturing fully the memory effects of the bath~\cite{Breuer2016,deVega2017}. The reduced dynamics of the OQS is  obtained from the EDMs through $\rho(t)=\delta^{\mathbf 0}_{{\boldsymbol \Phi}^{(t)}}\rho^{{\boldsymbol \Phi}^{(t)}}$ (by terminating the upward open arms with the closing tensor $\delta^0_\phi$, as shown in Fig.~\ref{fig:EDM_evolution}a).  Note that the number of EDMs and hence the corresponding number of time-local equations still increases exponentially with the evolution time.

The EDMs have a physical meaning as the state of an OQS evolving under a sequence of weak interventions.
Using the formula $\mathcal{S}^{\phi_t} \mathcal{S}^{\phi'_t}=P^{\phi_{t},\phi'_{t}}_{\phi''_{t}} \mathcal{S}^{\phi''_{t}}+O\left(\epsilon^2\right)$, we can demonstrate that the EDM tensor can be approximated as  [see Fig.~\ref{fig:EDM_evolution}b for illustration by tensor-network diagrams]
\begin{align}
    \rho^{\bm{\Phi}^{\prime(t)}} & \approx
    {\mathbb C}_{{\boldsymbol\Phi}^{(t)}} {\mathcal S}^{\phi'_t} {\mathcal S}^{\phi_t}\cdots{\mathcal S}^{\phi'_2}  {\mathcal S}^{\phi_2}{\mathcal S}^{\phi'_1}  {\mathcal S}^{\phi_1}\rho(0)
    \nonumber \\
    & =\text{Tr}_{\rm B}\left[\mathcal{S}^{\phi'_t} e^{\epsilon \mathcal{H}^{-}(t)}\cdots \mathcal{S}^{\phi'_1} e^{\epsilon \mathcal{H}^{-}(\epsilon)}\Omega(0)\right]\rho(0).
    \label{eq:EDM-phy}
\end{align}
Therefore, the trace of the EDMs corresponds to different types and orders of the time-ordered correlation functions of the OQS defined by the sequence of super-operators ${\mathcal S}^{\phi'_1}, {\mathcal S}^{\phi'_2}, \ldots, {\mathcal S}^{\phi'_t}$.

\section{Complexity of non-Markovian dynamics}\label{sec-PCT}

\subsection{Polynomial complexity theorem}

The EDMs offer a general formulation of OQS dynamics using cumulants of baths, but solving the dynamics is still challenging since the number of EDMs increases exponentially with the evolution time. Here, we prove a polynomial complexity (PC) theorem, which establishes that the number of independent equations needed to fully describe the dynamics of an OQS increases at most linearly with the evolution time and polynomially with the bath size. 
\begin{theorem}[Polynomial complexity theorem]\label{theorem:PC}
For a central quantum system (S)  coupled to a generic quantum bath (B), the EDMs that describe the non-Markovian dynamics of the system S in a finite time $T$ can be expanded by a set of linearly independent vectors, and the size of the set increases at most linearly with the evolution time, i.e.,
\begin{equation}\label{eq:FDE}
N_{T} \sim d^2 T/{\epsilon},
\end{equation}
where $d$ is the dimension of the Hilbert space of the system S and $\epsilon$ represents the minimal time in which the bath B induces a discernible change in the reduced dynamics of the system S.  
\end{theorem}
\begin{proof}
We write the evolution path of an EDM as an operator-valued vector ${\boldsymbol \rho}^{{\boldsymbol\Phi}^{(T)}}\equiv \left[{\rho}^{{\boldsymbol\Phi}^{(T)}}, \ldots, {\rho}^{{\boldsymbol\Phi}^{(2)}}, {\rho}^{{\boldsymbol\Phi}^{(1)}}\right]$.
Since each element ${\rho}^{{\boldsymbol\Phi}^{(t)}}$ in the vector is an Hermitian matrix and can be expanded by $d^2$ Hermitian system operators with $d^2$ independent real parameters (including the trace), each operator-valued EDM vector ${\boldsymbol \rho}^{{\boldsymbol\Phi}^{(T)}}$ can be expressed as a vector in a linear space of real dimension $T \times d^2$. Consequently, at most $N_T=T d^2$ independent  vectors can be formed using all the EDMs, given by
\begin{equation}\label{eq:linearTransformation}
    {\mathbf R}^{{a}} = \left[r\right]^{{a}}_{{{\boldsymbol \Phi}^{(T)}}} {\boldsymbol \rho}^{{\boldsymbol \Phi}^{(T)}},
\end{equation}
where ${\mathbf R}^{{a}} = \left[R_T^{{a}},\ldots,R^{{a}}_2,R^{{a}}_1\right]$ with each component $R^{{a}}_t$ being a $d$-dimensional Hermitian matrix at time $t$, $\left[r\right]^{{a}}_{{{\boldsymbol \Phi}^{(T)}}}$ are real coefficients for ${a}\in\left\{0,1,\ldots,N_T-1\right\}$ (see Appendix~\ref{Appendix:EDMvectors} for detailed definitions), and the summation over the indices ${{\boldsymbol \Phi}^{(T)}}$ is implied. These $N_T$ independent EDM vectors $\left\{{\mathbf R}^{{a}}\right\}$ are sufficient to fully describe the non-Markovian dynamics in the finite time range $[0,T]$.
\end{proof}

Theorem~\ref{theorem:PC} means that the ultimate complexity of OQS problems increases at most as a polynomial function of the evolution time. Note that the number of independent EDMs $N_T$ does not explicitly depend on the size of the bath. With increasing bath size, the energy range of the whole system (or, mathematically, the norm of the total Hamiltonian) may grow and therefore the time step $\epsilon$, which is inversely proportional to the norm of the Hamiltonian, needs to be  reduced. This means an implicit dependence of $N_T$ on the bath size. However, such a dependence would be only polynomial (instead of exponential), since the total energy increases with the bath size only polynomially (actually, at most linearly for local interactions).
Therefore, in principle, efficient solutions exist for general OQS problems.

\subsection{Efficient MPS representation of OQS dynamics}

The PC theorem upper-bounds the ultimate complexity of the OQS problems, but finding a set of independent EDM vectors $\left\{{\mathbf R}^{{a}}\right\}$ (and hence an explicit polynomial algorithm) is still challenging. Here, we propose an efficient tensor contraction algorithm to solve the EDM tensor. This algorithm is based on the matrix product state (MPS) representation of the EDMs $\rho^{{\boldsymbol \Phi}^{(t)}}$, in which the correlation tensors ${\mathbb{C}}^{{\boldsymbol \Phi}^{(t)}}_{{\boldsymbol \Phi}^{'(t)}}$ (and hence the combined kernel tensors ${\mathbb{K}}^{\phi_t;{\boldsymbol \Phi}^{(t-1)}}_{\phi'_t;{\boldsymbol \Phi}^{'(t-1)}}$) are {\em effectively} represented by matrix product operator (MPO) tensor networks~\cite{Orus2014}.  

\begin{figure}[htpb]
    \centering
    \includegraphics[width=0.85\linewidth]{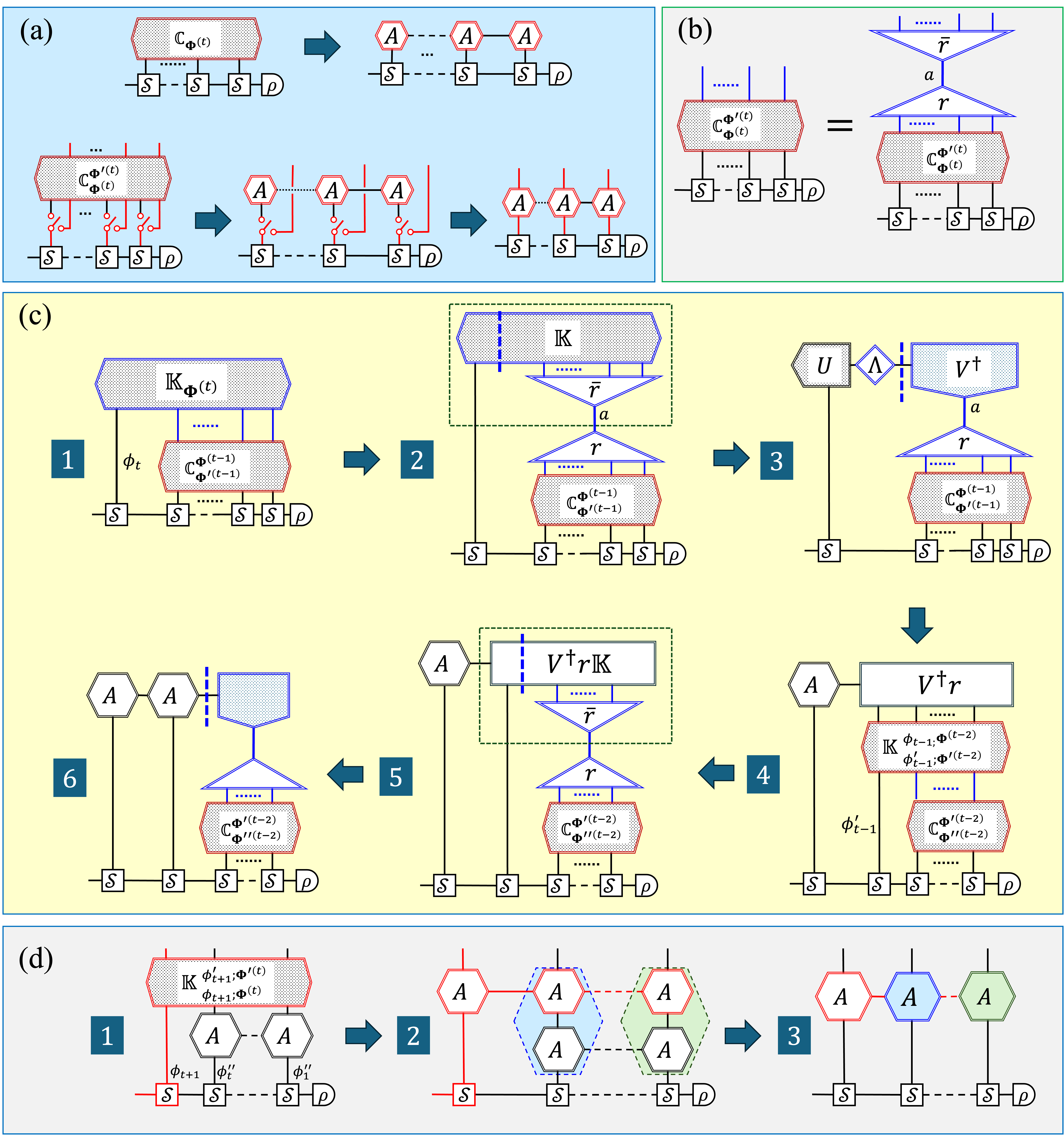}
    \caption{MPS representation of EDMs. (a) Theorem~\ref{theorem:MPS} represented in tensor-network diagrams. (b) Transform between EDMs $\rho^{{\boldsymbol\Phi}^{\prime (t)}}$ and component $R_t^a$ of independent EDM vectors by isometries $r$ and $\bar{r}$. (c) Tensor-network diagrams showing how an EDM can be decomposed into an MPS with bond dimension between local tensors upper-bounded by the number of independent EDM vectors (the range of $a$-index of the isometry in diagram 2), which proves Theorem~\ref{theorem:MPS}. (d) Algorithm for recursive construction of the MPS representation of an EDM.}
    \label{fig:Theorem_MPS}
\end{figure}

\begin{theorem}[{linearly increasing bond dimension theorem}]\label{theorem:MPS} The density matrix of an OQS has a {\em time-local} MPS representation as (see Fig.~\ref{fig:Theorem_MPS}a)
\begin{equation}\label{eq:MPS_DM}
\rho(t)=\left[{A}_{\phi_t}\right]_{a_{t-1}}\left[{A}_{\phi_{t-1}}\right]_{a_{t-1},a_{t-2}} \cdots \left[{A}_{\phi_2}\right]_{a_2,a_1} \left[{A}_{\phi_1}\right]_{a_1}{{\mathbb S}^{{\boldsymbol \Phi}^{(t)}}}\rho(0),
\end{equation}
and in turn (with ${A}^{\phi_t}_{\phi'_t}\equiv {A}_{\phi''_t} P^{\phi_t,\phi''_t}_{\phi'_t}$) the EDM
\begin{equation}
   \label{eq:MPS_EDM} 
   \rho^{{\boldsymbol \Phi}^{(t)}}=\left[{A}^{\phi_t}_{\phi'_t}\right]_{a_{t-1}}\left[{A}^{\phi_{t-1}}_{\phi'_{t-1}}\right]_{a_{t-1},a_{t-2}} \cdots \left[{A}^{\phi_2}_{\phi'_2}\right]_{a_2,a_1} \left[{A}^{\phi_1}_{\phi'_1}\right]_{a_1}
   {{\mathbb S}^{{\boldsymbol \Phi}^{'(t)}}}\rho(0),
    \end{equation}
and the MPS representation is {\em efficient} in the sense that the dimension $D_\tau$ (defined as the range of the lateral index $a_\tau$) of the bonds between the local tensors ${A}$'s increases at most linearly in time, that is
\begin{equation}\label{eq:upperbound-BD}
D_\tau \le N_{\tau}\sim d^2 \tau/\epsilon.
\end{equation}
\end{theorem}
\begin{proof}
Figure~\ref{fig:Theorem_MPS}b and \ref{fig:Theorem_MPS}c illustrate the proof. 
The EDMs can be obtained from the independent EDM vectors by 
${\boldsymbol \rho}^{{\boldsymbol \Phi}^{(T)}}=\left[\bar{r}\right]^{{\boldsymbol\Phi}^{(T)}}_{a} {\mathbf R}^{{a}}$, where the inverse transform $\bar{r}$ is defined by $\left[r\right]_{{\boldsymbol\Phi}^{(T)}}^{a}\left[\bar{r}\right]^{{\boldsymbol\Phi}^{(T)}}_{{b}}=\delta^{{a}}_{{b}}$. Equation~\eqref{eq:linearTransformation} holds for each component of the vectors. Therefore, 
\begin{equation}
    \label{eq:isometry}
%\begin{align}
    R^a_{{t}}  =\left[r\right]^{{a}}_{{{\boldsymbol \Phi}^{({t})}}} {\rho}^{{\boldsymbol \Phi}^{({t})}},
\ \ {\rm and}\ \ 
{\rho}^{{\boldsymbol \Phi}^{({t})}}  = \left[\bar{r}\right]_{{a}}^{{{\boldsymbol \Phi}^{({t})}}}  R^a_{{t}}.
%\end{align}
\end{equation}
This means the transforms $r$ and $\bar{r}$ can be viewed as isometry tensors as shown in Fig.~\ref{fig:Theorem_MPS}b.
Consider the density {matrix} $\rho(t)$, which can be obtained by applying the kernel tensor ${\mathbb K}_{{\boldsymbol \Phi}^{(t)}}\equiv 
{\mathbb K}^{0;{\mathbf 0}}_{\phi_t;{\boldsymbol \Phi}^{(t-1)}}$ to the EDM $\rho^{{\boldsymbol \Phi}^{(t-1)}}$, as shown in diagram 1 in Fig.~\ref{fig:Theorem_MPS}c. By Theorem~\ref{theorem:PC}, the EDM can be reduced to $R_{t-1}^a$ using the isometry $r$ and then be recovered by $\bar{r}$ (diagram 2 in Fig.~\ref{fig:Theorem_MPS}c), with the bond dimension between the two isometries $r$ and $\bar{r}$ (the range of $a$) being $D_{t-1}\le N_{t-1}$. Then one can perform the singular value decomposition (SVD) of the combined tensor ${\mathbb K}\bar{r}$ as indicated by the vertical blue dashed line in the diagram, with the result ${\mathbb K}\bar{r}=U\Lambda V^{\dag}$, in which $U$ and $V^{\dag}$ are semi-unitary and the positive-definite diagonal matrix $\Lambda$ contains all non-zero singular values  (see diagram 3)~\cite{Schollwock2011}. The number of non-zero singular values and hence the dimension of the bond between $\Lambda$ and $V^{\dag}$ are at most $D_{t-1}$.
The local tensor can be obtained as ${A}\equiv U\Lambda$ (diagram 4). The combined tensor $V^{\dag}r{\mathbb K}$ applied to the EDM at $t-2$ (diagram 5) can be separated using the same procedure as that transforming diagram 2 to 3 to obtain another local tensor ${A}$ as shown in diagram 6. This procedure is continued until all local tensors ${A}$'s are obtained. Therefore, Eq.~\eqref{eq:MPS_DM} is obtained (tensor diagram shown in upper row of Fig.~\ref{fig:Theorem_MPS}a) and the upper limit of bond dimension in Eq.~\eqref{eq:upperbound-BD} is guaranteed. By applying the picking tensors (illustrated in lower row of Fig.~\ref{fig:Theorem_MPS}a), Eq.~\eqref{eq:MPS_EDM} is also proved.
\end{proof}
\noindent{\bf Remarks}:
\begin{enumerate}
    \item As seen from the proof, the effects of a correlation tensor on the dynamics of the OQS are equivalent to those of a  {\em time-local} MPO, that is
\begin{equation}\label{eq:MPO_C}{\mathbb{C}}_{{\boldsymbol \Phi}^{'(t)}} \sim \left[{A}_{\phi'_t}\right]_{a_{t-1}}\left[{A}_{\phi'_{t-1}}\right]_{a_{t-1},a_{t-2}} \cdots \left[{A}_{\phi'_2}\right]_{a_2,a_1} \left[{A}_{\phi'_1}\right]_{a_1},
\end{equation}
and in turn 
\begin{equation}
   \label{eq:MPO_C_ud} 
    \mathbb{C}^{{\boldsymbol \Phi}^{(t)}}_{{\boldsymbol \Phi}^{'(t)}}\sim \left[{A}^{\phi_t}_{\phi'_t}\right]_{a_{t-1}}\left[{A}^{\phi_{t-1}}_{\phi'_{t-1}}\right]_{a_{t-1},a_{t-2}} \cdots \left[{A}^{\phi_2}_{\phi'_2}\right]_{a_2,a_1} \left[{A}^{\phi_1}_{\phi'_1}\right]_{a_1}.
    \end{equation}
 The MPO representation dramatically reduces the complexity of the correlation tensors. Although the correlations of fluctuations in a many-body environment can be very complex (requiring $2^{{\mathcal O}(L)}$ free parameters to describe all $L$-th order correlations), they need to go through a very much limited number of channels to affect the dynamics of an OQS since the central system has a finite number of degrees of freedom, and therefore they can be organized into a small number of ``modes'' represented by an MPO.
 It should be noted that the equivalence in Eqs.~\eqref{eq:MPO_C} and \eqref{eq:MPO_C_ud} is only {\em effective}, that is, it does not mean that in general the correlation tensors have efficient MPO representation.
\item The combined kernel tensor ${\mathbb{K}}^{\phi_t;{\boldsymbol \Phi}^{(t-1)}}_{\phi'_t;{\boldsymbol \Phi}^{'(t-1)}}$ can be regarded as a special case of a correlation tensor (${\mathbb{K}}^{\phi_t;{\boldsymbol \Phi}^{(t-1)}}_{\phi'_t;{\boldsymbol \Phi}^{'(t-1)}}={\mathbb{C}}^{{\boldsymbol \Phi}^{(t)}}_{{\boldsymbol \Phi}^{'(t)}}$ when ${\mathbb{C}}^{{\boldsymbol \Phi}^{(t-1)}}_{{\boldsymbol \Phi}^{'(t-1)}}$ is an identity tensor), so it also has an efficient MPO representation when its effects on the OQS are considered.
\item  The {\em time-local} MPO structure of the correlation tensors does not mean that the existence of a time-local equation of motion for the OQS dynamics, since the local tensors ${A}$'s need to be updated for each step of time evolution. Indeed, the evolution equation for EDMs in Eq.~\eqref{eq:EDM-Evolution} is non-local in time.
\end{enumerate}

Theorem~\ref{theorem:MPS} does not assume any constraint on the memory time in the non-Markovian process. When the bath has a finite memory time, the bond dimension in the MPS representation is reduced according to the following corollary.

\begin{corollary}\label{Cor_MPS_memory}
When the bath has a finite memory time $T_m$, i.e., if all the culmulants $\tilde{C}_{\phi_{t_l}, \ldots,\phi_{t_2},\phi_{t_1}}=0$ for $t_l-t_1>T_m$, the bond dimension between the local tensors ${A}$'s in Eqs.~\eqref{eq:MPS_DM} and \eqref{eq:MPS_EDM} is upper-bounded by
\begin{equation}\label{eq:PC-memory}
D_{\tau} \le  \min(N_{\tau},N_{T_m}).
\end{equation}
\end{corollary}
 This is because the memory time limits the number of independent EDM vectors involved in the evolution [see Appendix~\ref{sec:FBDT-FMT-app} for a proof].

Theorem~\ref{theorem:MPS} and Corollary~\ref{Cor_MPS_memory} indicate that, once the cumulants of the noises from the bath are known (by computation or measurement) or if the OQS evolution has a tensor network representation, which are the case in a broad range of models (with examples discussed in Sec.~\ref{Sec:applications}), one can recursively solve the EDM tensor with the computation hour and memory expenditure increasing polynomially with the evolution time and the bath size. In the proof of Theorem~\ref{theorem:MPS}, the MPS representation has been obtained for each time $t$ from the end to the beginning of the evolution. To save the computational time consumption, the local tensors ${A}$ can actually be recursively calculated from the beginning to the end of evolution as illustrated in Fig.~\ref{fig:Theorem_MPS}d. The procedure is as follows. First, the calculation of ${A}$ tensor for $t=1$ is trivial. Then, if the ${A}$ tensors have been obtained for $t$, the EDM for $t+1$ is reached by applying the combined kernel function ${\mathbb K}^{\phi'_{t+1};{\boldsymbol \Phi}^{'(t)}}_{\phi_{t+1};{\boldsymbol \Phi}^{(t)}}{\mathcal S}^{\phi_{t+1}}$ (see diagram 1 in Fig.~\ref{fig:Theorem_MPS}d). According to Theorem~\ref{theorem:MPS}, the ${\mathbb K}$ tensor has an efficient MPO representation as shown in diagram 2 in Fig.~\ref{fig:Theorem_MPS}d. The time-local ${A}$ tensors at the same time (enclosed by dashed hexagons in diagram 2 in Fig.~\ref{fig:Theorem_MPS}d) can be combined into a single ${A}$ tensor as shown in diagram 3 in Fig.~\ref{fig:Theorem_MPS}d. The ${A}$ tensors in diagram 3 in Fig.~\ref{fig:Theorem_MPS}d can be reduced by transforming the tensor network into the right canonical form~\cite{Schollwock2011} and Corollary~\ref{Cor_MPS_memory} guarantees that the bonds between them have dimensions bounded by Eq.~\eqref{eq:PC-memory}. This procedure can be repeated until the EDM at a desired time $T$ is obtained.

Note that other numerical tensor-network methods have been proposed to deal with non-Markovian dynamics, such as the time-evolving MPO (TEMPO)~\cite{Strathearn2018} and the process tensor in the MPO (PT-MPO) form~\cite{Pollock2018,Jorgensen2019,Link2024}. The TEMPO method is designed to deal with OQS's in Gaussian bosonic environments. The PT-MPO method is formulated for general OQS's by solving the process tensor, which in general, as noted in Ref.~\cite{Aloisio2023}, is classically hard due to the potentially divergent bond dimension. In contrast, the method presented in this paper solves the EDM tensor, whose bond dimension is guaranteed to increase only linearly with evolution time.

\section{Application of EDM tensor networks}
\label{Sec:applications}
\subsection{Spin-boson model}
\label{Sec_spin-boson}
Here we use the spin-boson model~\cite{Leggett1987} to demonstrate  the efficiency of the EDM method. The Hamiltonian in the interaction picture reads 
\begin{equation}
	 {H}(t)= {S}_{z}(t)\sum_{k} g_{k} \left({a}^{\dagger}_{k}e^{i\omega_k t}+{\rm h.c.}\right)\equiv {S}_{z}(t) B(t),
\end{equation}
where ${S}_{z}(t)=\cos \left(\mu t\right) {S}_{z}+\sin \left(\mu t\right) {S}_{y}$ describes the $z$-direction spin-1/2 operator under the tunneling term $\mu {S}_{x}$, ${a}^{\dagger}_{k}$ 
 (${a}_{k}$) is the creation (annihilation) operator of the $k$-th bosonic mode with frequency $\omega_{k}$, and $g_{k}$ is the coupling strength between the spin and the $k$-th bosonic mode.  The distribution of the coupling strengths is described by the spectral density $J(\omega)\equiv \sum_{k} g^{2}_{k} \delta(\omega-\omega_{k})$. In the numerical calculation, we set the initial state of the system to be $\rho(0)= {S}_{z}+1/2$ (fully polarized along the $z$-axis) and the environment in the zero-temperature state (vacuum), i.e. $\Omega=\left|{\rm vacuum}\right\rangle\left\langle{\rm vacuum}\right|$ and take the Ohmic spectrum $J(\omega)=2 {J_0} \omega e^{-\omega/\omega_c}$ with a dimensionless coupling constant ${J_0}$ and a cutoff frequency $\omega_c$.  The quantum noise $B(t)$ has only second-order cumulants. From $J(\omega)$ we analytically obtain the combined kernel tensor and then numerically evaluate the EDM tensor using Eq.~\eqref{eq:EDM-Evolution} and the tensor-network contraction procedure shown in Fig.~\ref{fig:Theorem_MPS}d [see Appendix~\ref{sec:SBM-app} for details].

\begin{figure}[thpb]
	\centering
	\includegraphics[width=0.9\columnwidth]{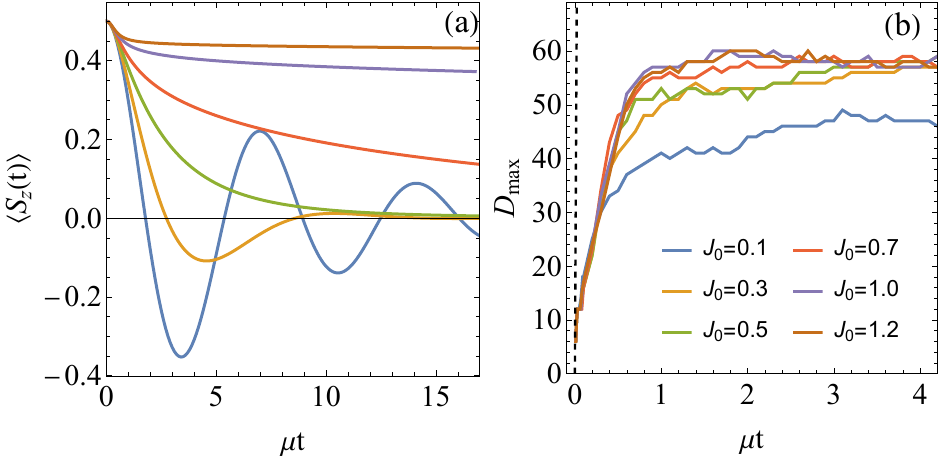}
	\caption{Solution of the reduced spin dynamics in the spin-boson model. (a) Spin polarization along the $z$ axis as a function of time for various spin-bath coupling strengths $J_0$. (b) The maximum bond dimension as a function of the evolution time $t$ for truncation precision $\xi=10^{-5}$. The dashed line is the upper bound given by Corollary~\ref{Cor_MPS_memory}. Here, the evolution step ${\epsilon}=0.01 \mu^{-1}$ and the noise frequency cutoff $\omega_c=5.0\mu$.}
	\label{fig:SBM}
\end{figure}

The spin polarization $\left\langle S_{z}(t)\right\rangle=\text{Tr}\left[{S}_{z}(t) \rho(t)\right]$ is illustrated in Fig.~\ref{fig:SBM}a for various coupling strengths $J_0$. In the calculation, we set the time step as ${\epsilon}=0.01 \mu^{-1}$ and introduce a self-adaptive truncation of the bond dimension by  discarding the singular values $\{\lambda_{\alpha}\}$ that satisfy $\lambda_{\alpha}/\lambda_{d^2+1} \le \xi$ with the truncation precision $\xi$ being a small number (e.g., $10^{-5}$). The spin relaxation dynamics exhibits a crossover from oscillation to overdamped decay at $J_0\approx 0.5$, indicating the competition between the coherent evolution and dissipation~\cite{Leggett1987, Strathearn2018}. Figure~\ref{fig:SBM}b shows the maximum bond dimension $D_{\rm max}\equiv \max\left( D_{\tau}\right)$ increasing with the evolution time $t$. In short times, $D_{\rm max}$ increases linearly for all coupling strengths, which means that the number of EDMs involved in the relaxation dynamics increases linearly. However, $D_{\rm max}$ saturates at later times, suggesting that a finite number of EDMs is sufficient for describing the non-Markovian dynamics. The maximum bond dimension increases faster and saturates at higher values for stronger coupling, which implies a dependence on the bath size since $J_0$ is proportional to the density of bosonic modes. However, the dependence of the bond dimension on $J_0$ (or the mode density) is sublinear (with a saturating behavior for strong coupling). These results validate the upper bound of the bond dimension given by Eq.~\eqref{eq:PC-memory} (dashed line in Fig.~\ref{fig:SBM}b). In addition, we verify that the linearly increasing bond dimension theorem holds regardless of the truncation precision [see Appendix~\ref{sec:SBM-app} for details]. 

\subsection{Separable baths}
\label{Sec:ICD}

\subsubsection{General formalism}

There are cases where the quantum baths are separable, that is, can be decomposed into (approximately) independent sub-baths, as illustrated in Fig.~\ref{fig:ICD}a. Examples include a bath composed of independent bosonic modes~\cite{Leggett1987}, distinct spins in a spin bath~\cite{Hanson2007, Khaetskii2002}, and multi-bath OQS's~\cite{Chan2014,Giulio2017}. In such cases, the noise fields are sums of those from different sub-baths, $B_{\alpha}(t)=\sum^{K}_{k=1} B_{k;\alpha}(t)$, with $K$ being the number of sub-baths.
The bath being separable means that the noises from different sub-baths have no correlations, or mathematically, $\left[B_{k;\alpha}(t),B_{k';\alpha'}(t')\right]=0$ for $k\ne k'$ and $\Omega=\otimes_{k=1}^K \Omega_k$  with $\Omega_k$ being the initial density matrix of the $k$-th sub-bath. For such separable baths, the cumulants can be decomposed as
\begin{equation}\label{eq:componentDecomposition}
\tilde{C}_{{\boldsymbol\Theta}}=\sum^{K}_{k=1} \tilde{C}_{k;{\boldsymbol\Theta}},
\end{equation}
where $\tilde{C}_{k;{\boldsymbol\Theta}}$ denotes the cumulant of the $k$-th sub-bath. This decomposition leads to a recursive formula of the EDMs [tensor diagrams illustrated in Fig.~\ref{fig:ICD}(b, c)]
\begin{equation}\label{eq:Gaudin-recursive}
    \rho^{{\boldsymbol \Phi}^{(T)}}_{L+1}=  \mathbb{C}_{L+1; {\boldsymbol \Phi}^{'(T)} }\left[\prod^{T}_{\tau=1} P^{\phi_\tau, \phi'_\tau}_{\phi''_{\tau}}\right] \rho^{{\boldsymbol \Phi}^{''(T)}}_{L}
,
\end{equation}
where $\rho^{{\boldsymbol \Phi}^{''(T)}}_{L}$ denotes the EDM tensor of the system coupled to the first $L$ sub-baths with $\rho^{{\boldsymbol \Phi}^{(T)}}_{0}= {\mathcal S}^{{\boldsymbol\Phi}^{(T)}}\rho(0)$.  Using the cumulants or the combined kernel functions of the $(L+1)$-th sub-bath ${\mathbb K}_{L+1;{\boldsymbol\Phi}^{(t)}}^{{\boldsymbol\Phi}^{'(t)}}$ and the EDM $\rho^{{\boldsymbol \Phi}^{(t)}}_{L}$, one can calculate recursively the EDMs in a series of sub-baths following the same procedure as in Fig.~\ref{fig:Theorem_MPS}d. Note that the linearly increasing bond dimension theorem holds for every EDM tensor $\rho^{{\boldsymbol \Phi}^{(T)}}_{L}$.

\begin{figure}
    \centering
    \includegraphics[width=0.9\linewidth]{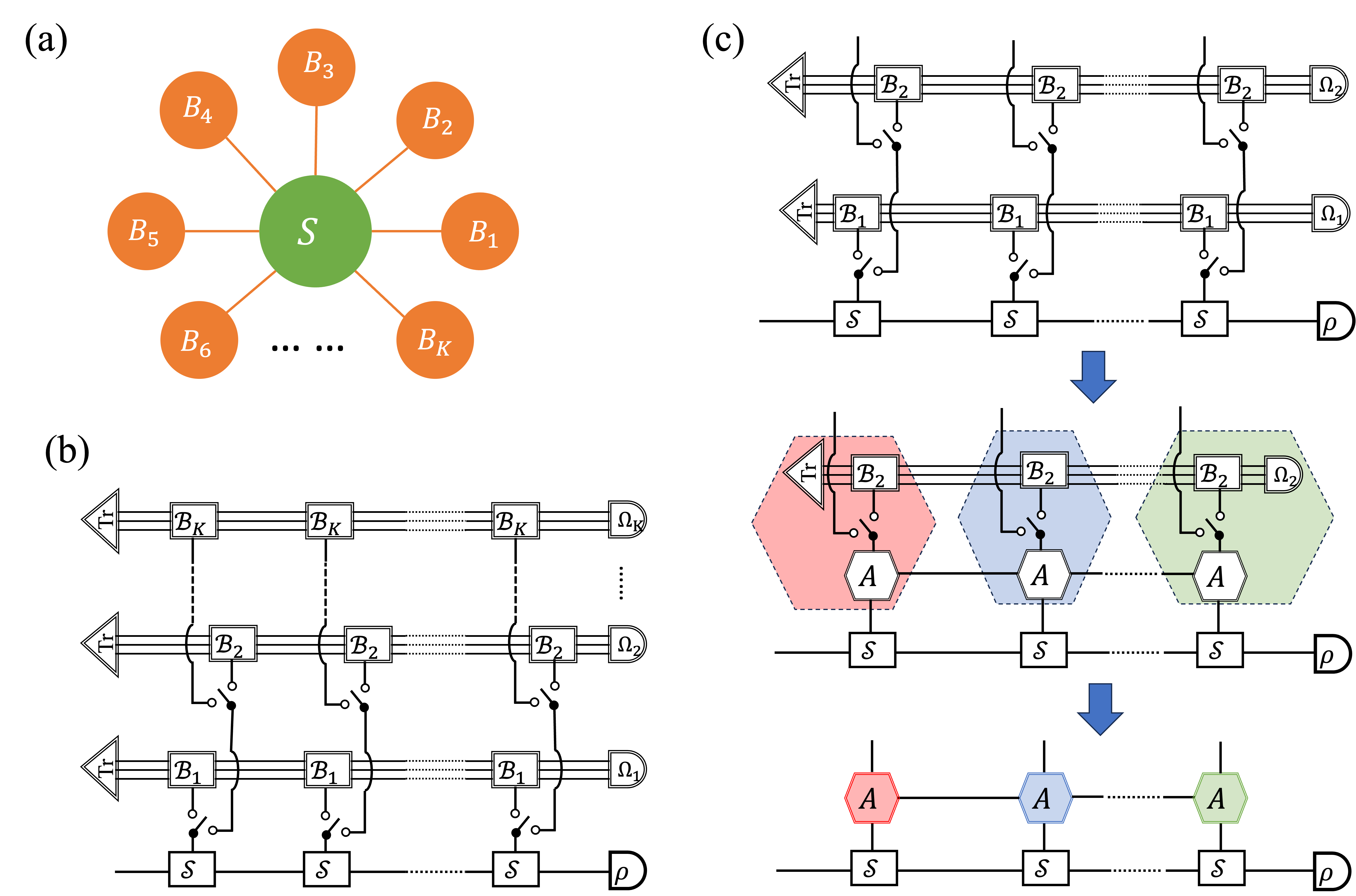}
    \caption{OQS dynamics in a separable bath. (a) Schematic of a central system coupled to $K$ independent sub-baths. (b) The decomposition of the correlation tensor into contributions from different sub-baths. (c) The contraction procedure shown by an example of two sub-baths. A series of EDMs $\rho^{{\boldsymbol \Phi}^{(t)}}_{L}$ for $L=1,2,\ldots, K$ are obtained recursively.}
    \label{fig:ICD}
\end{figure}

\subsubsection{Gaudin model}
As a typical and important example of separable baths, we consider the Gaudin model~\cite{Gaudin1976}, in which a central spin interacts with $K$ independent bath spins. The Hamiltonian reads
\begin{equation}\label{eq:GaudinModel}
     {H}(t)=\sum^{K}_{k=1} g_k  {{\mathbf S}} \cdot  {{\mathbf J}}_{k},
\end{equation}
where $g_k$ is the strength of the coupling between the central spin-1/2 ${{\mathbf S}}\equiv ({S}_{x}, {S}_{y}, {S}_{z})$ and the $k$-th bath spin-1/2 ${{\mathbf J}}_{k}\equiv({J}_{k;x}, {J}_{k;y}, {J}_{k;z})$. The solution of the Gaudin model is in general hard. Solutions are available only for some special cases, such as reduction of the bath to single large spins when the coupling is uniform ($g_k={\rm constant}$), exact diagonalization for small $K$~\cite{Cywifmmode2010}, Bethe ansatz for highly polarized spin bath~\cite{Bortz2007, He2022}, and density matrix renormalization group method for short time dynamics~\cite{Stanek2013, Stanek2014}.

\begin{figure}[thpb]
    \centering
    \includegraphics[width=0.9\columnwidth]{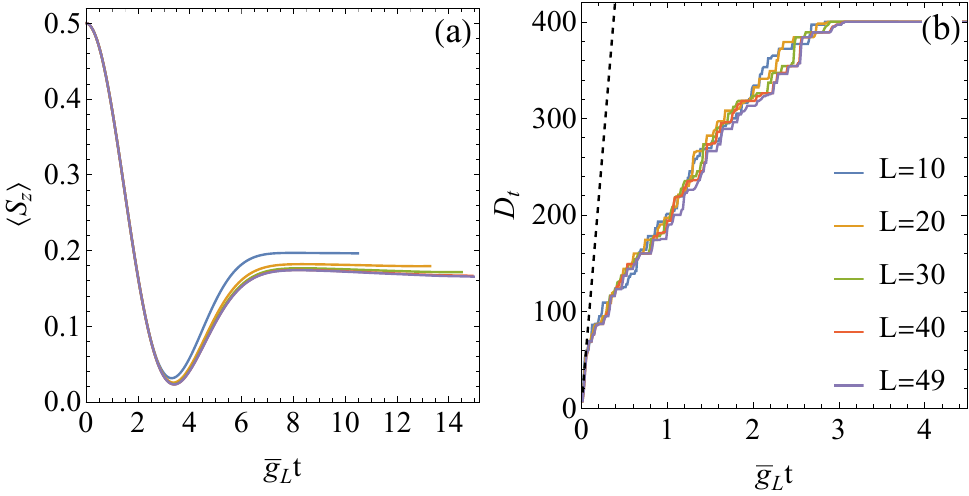}
    \caption{Solution of Gaudin model. (a) The central spin polarization with the first $L$ spins included in the bath. (b) The bond dimension $D_t$ in evaluating the EDM tensor $\rho^{{\boldsymbol \Phi}^{T}}_{L}$ with $t=\epsilon, \cdots, T$. The dashed line is the upper bound given by Corollary~\ref{Cor_MPS_memory}. The time is scaled by $\bar{g}_L$, the effective coupling strength between system and the $L$-spin bath. The plateau is the bond dimension cutoff ($D_c=400$) set in the numerical calculation. Parameters are $N=49$, $\epsilon=0.03 g^{-1}$, and $T=15.0g^{-1}$ with $g$ being the coupling constant. The truncation precision is set as $\xi=10^{-6}$.}
    \label{fig:Gaudin}
\end{figure}

An important feature of the Gaudin model is that the bath spins are time-independent, which leads to an infinite memory time. Furthermore, each bath spin constitutes a highly non-Gaussian quantum noise with an infinite number of high-order cumulants, which makes the cumulant expansion infeasible. Nevertheless, the  decomposition procedures shown in Eq.~\eqref{eq:Gaudin-recursive} offer an efficient approach to solving the EDM tensor. The procedure is as follows (as illustrated in Fig.~\ref{fig:ICD}c). First, we compute the MPS EDM of the central spin $\rho^{{\boldsymbol \Phi}^{(T)}}_{L=1}$ with only the first bath spin considered. Then, using $\rho^{{\boldsymbol \Phi}^{(T)}}_{L=1}$ from the previous step, we compute the MPS EDM of the central spin including the effects of the first 2 bath spins, $\rho^{{\boldsymbol \Phi}^{(T)}}_{L=2}$. This is repeated until all bath spins are included and the MPS EDM $\rho^{{\boldsymbol \Phi}^{(T)}}=\rho^{{\boldsymbol \Phi}^{(T)}}_{L=K}$ is obtained. Each step is essentially an exactly solvable single-spin bath problem [see Appendix~\ref{sec:GM-app} for details].

In the numerical calculation, we assumed that the central spin is initially polarized along the $z$ axis, and that the bath spins are at an infinitely high temperature, that is, unpolarized. The couplings are set to have a linearly decreasing distribution $g_k=g \sqrt{\frac{6K}{2K^2+3K+1}} \frac{K+1-k}{K}$~\cite{Stanek2013}. The numerical results are shown in Fig.~\ref{fig:Gaudin}.  In the calculation, the bond dimension $D_t$ is truncated by discarding the singular values $\{\lambda_{\alpha}\}$ that satisfy $\lambda_{\alpha}/\lambda_{d^2+1} \le \xi$ with the truncation precision $\xi=10^{-6}$. In addition, a hard cutoff $D_c=400$ of the bond dimension is introduced to manage the computational time. The error introduced by this cutoff is analyzed in Appendix~\ref{sec:GM-app}. As shown in Fig.~\ref{fig:Gaudin}b, the bond dimension $D_{t}$ in the EDM tensor $\rho^{{\boldsymbol \Phi}^{(T)}}_{L}$ in the bath including the strongest $L$ spins increases linearly with time $t$  until it reaches the cutoff, which demonstrates the linearly increasing bond dimension theorem. It is interesting to note that the bond dimension increases nearly independently of the bath size as a function of the scaled time $\bar{g}_Lt$, where the effective coupling strength $\bar{g}_L\equiv \sqrt{\sum^{L}_{k=1} g^2_k}$ for the $L$-spin bath (see Fig.~\ref{fig:Gaudin}b). Appendix~\ref{Append_universal} presents more data to demonstrate the universal behavior of bond dimension increasing with the scaled time. This universal behavior means that for a given real time $t$ the number of independent equations needed to describe the central spin dynamics increases at most linearly with the bath size.

\subsubsection{Possible generalization to inseparable baths}
The formalism presented here for computing EDMs in separable baths where the sub-baths have no interactions or initial correlations  can be generalized to non-separable baths as long as the sub-baths can be arranged in a one-dimensional chain or tree-like structure and the interactions and initial correlations, if exist, are only {\em short-ranged}. The possibility of such generalization results from the chain-structure feature of Eq.~\eqref{eq:Gaudin-recursive}  as shown in Fig.~\ref{fig:ICD}b. To generalize the EDM tensor network, one needs to modify the picking tensors connecting the sub-baths to be MPOs taking into account the short-range interactions or initial correlations. Remarkably, even when solving the chain-like quantum bath itself is computationally hard, the PC theorem guarantees that the OQS problem is solvable.

\section{Discussions and conclusions}
\label{Sec_discussion}
In conclusion, the problem of non-Markovian dynamics of OQS's in general is {\em not} computationally hard according to the polynomial complexity theorem. Furthermore, the problem can be formulated in terms of extended density matrix tensor networks. Guaranteed by the  linearly increasing bond dimension theorem, the computational complexity of evaluating the EDM tensors increases only as a polynomial function of the evolution time and the bath size. We demonstrated the theorems and the method using two typical models, namely, the spin-boson model and the  Gaudin model. The theory and numerical examples show that the OQS problems are solvable as long as the cumulants of noises in baths are known (by measurement or computation).  Several numerical methods are available for computing the noise correlations and cumulants, such as the link-cluster expansion~\cite{Saikin2007} and cluster-correlation expansion~\cite{Yang2008, Yang2009, Witzel2014}.  It is also possible to experimentally measure cumulants for many complex quantum baths through, e.g., quantum nonlinear spectroscopy~\cite{Wang2019, Meinel2022}. Another scenario of existence of explicit efficient algorithms for OQS problems is the cases where a two-dimensional tensor network can be constructed for the EDMs (including the temporal dimension for evolution), such as separable baths and one-dimensional baths with short-range interactions.

The significance of the PC theorem goes beyond the EDM tensor-network formalism, because the theorem guarantees the existence of an efficient solution of the OQS problem regardless of the complexity of the quantum many-body bath. The question is how to find an explicit algorithm.
In the EDM formalism, the noise cumulants are needed, but the computation of the noise correlations in quantum many-body baths is in general a difficult problem~\cite{Aloisio2023}. 
However, in many cases, the OQS problem can be formulated without explicitly involving the noise correlations in the baths. For example, bath states can be formulated as tensor networks for a broad range of many-body systems~\cite{Cirac2021, Orus2019} and the reduced dynamics of the OQS can be obtained by tensor contraction, which, if executed in a proper procedure, is guaranteed by the PC theorem to have no exponentially growing bond dimension.

The PC theorem and the EDM formalism also offer new approaches to  studying many-body systems. The evaluation of local observables in quantum many-body systems can be formulated as equivalent to an OQS problem. A direct consequence is that the complexity of many-body problems can be largely reduced according to the PC theorem when only the local observables are of interest.

\acknowledgements{
This work was supported by the Innovation Program for Quantum Science and Technology Project No. 2023ZD0300600, the Hong Kong Research Grants Council - General Research Fund (Project No. 14302121), the National Natural Science Foundation of China/Hong Kong Research Council Collaborative Research Scheme Project CRS-CUHK401/22, the Hong Kong Research Grants Council Senior Research Fellow Scheme Project SRFS2223-4S01, and the New Cornerstone Science Foundation.
}

\appendix

\section{Tensor network representation of correlations} \label{Appd-sec:KT-app}
The correlation in Eq.~\eqref{eq:cumulantEx} is in the form of summation of cumulant decomposition. Here we transform the correlation tensor into a tensor product of kernel functions 
$ \tilde{K}_{{\boldsymbol\Theta}} \equiv \delta^{{\mathbf 0}}_{{\boldsymbol\Theta}}+\tilde{C}_{{\boldsymbol\Theta}}$ (an example tensor network shown in Fig.~\ref{fig:TensorRep}d).
Using  the picking tensors $ P^{\phi',\phi''}_{\phi}$ and ${\mathbb P}^{{\boldsymbol\Theta}',{\boldsymbol\Theta}''}_{{\boldsymbol\Theta}}\equiv\prod_{\phi_{\alpha}\in{\boldsymbol \Theta}}P^{\phi'_{\alpha},\phi''_{\alpha}}_{\phi_{\alpha}}$, we define the open-armed correlation tensors and kernel tensors as
\begin{subequations} 
  \begin{align}
{\mathbb C}^{{\boldsymbol\Phi}^{'(t)}}_{{\boldsymbol\Phi}^{(t)}} & \equiv {\mathbb P}^{{\boldsymbol\Phi}^{'(t)},{\boldsymbol\Phi}^{''(t)}}_{{\boldsymbol\Phi}^{(t)}}{\mathbb C}_{{\boldsymbol\Phi}^{''(t)}}  ,\label{eq:def_C_updwn}\\
\tilde{K}^{{\boldsymbol\Phi}^{'(t)}}_{{\boldsymbol\Phi}^{(t)}}
\left({\boldsymbol\Theta}\right) & \equiv  {\mathbb P}^{{\boldsymbol\Phi}^{'(t)},{\boldsymbol\Phi}^{''(t)}}_{{\boldsymbol\Phi}^{(t)}} \tilde{K}_{{\boldsymbol\Theta}''} {\delta}^{\mathbf 0}_{\bar{\boldsymbol\Theta}''} =\delta^{{\boldsymbol\Phi}^{'(t)}}_{{\boldsymbol\Phi}^{(t)}}+{\delta}^{\bar{\boldsymbol\Theta}'}_{\bar{\boldsymbol\Theta}} {\delta}^{{\boldsymbol\Theta}'}_{\mathbf 0} \tilde{C}_{{\boldsymbol\Theta}}
, \label{eq.:def_kernel_combined}
\end{align}
\end{subequations}
where ${\boldsymbol\Theta}\subseteq {\boldsymbol \Phi}^{(t)}$ is an index subset and $\bar{\boldsymbol\Theta}\equiv{\boldsymbol \Phi}^{(t)}\setminus{\boldsymbol\Theta}$ is the complement. Note that the tensor $\tilde{K}^{{\boldsymbol\Phi}^{'(t)}}_{{\boldsymbol\Phi}^{(t)}}
\left({\boldsymbol\Theta}\right)$ simply passes the indices not involved in the cumulant $\tilde{C}_{\boldsymbol\Theta}$ to the next level. By the meaning of the picking tensor or direct calculation, one can show that the picking tensor is symmetric, i.e., $P_{\phi}^{\phi',\phi''}=P_{\phi}^{\phi'',\phi'}$ and the product of two picking tensors is exchangeable, i.e.,
$$
P_{\varphi}^{\phi',\phi'''} P_{\phi}^{\varphi,\phi''}=P_{\phi}^{\phi''',\varphi}P_{\varphi}^{\phi'',\phi'},
$$
which is illustrated in Fig.~\ref{fig-appendix:product-picking}a.
Therefore, the tensor product of two kernel tensors is exchangeable, i.e.,
$\left[\tilde{K}\left({\boldsymbol\Theta}_1\right)\tilde{K}\left({\boldsymbol\Theta}_2\right)
\right]^{{\boldsymbol\Phi}^{'(t)}}_{{\boldsymbol\Phi}^{(t)}}=\left[\tilde{K}\left({\boldsymbol\Theta}_2\right)\tilde{K}\left({\boldsymbol\Theta}_1\right)\right]^{{\boldsymbol\Phi}^{'(t)}}_{{\boldsymbol\Phi}^{(t)}}$.
The combined kernel tensor that contains all kernel functions with the first index being $\phi_{t+1}$ is defined in the tensor product form as
\begin{equation}\label{eq:CombinedKT-app}
{\mathbb K}^{\phi'_{t+1};{\boldsymbol\Phi}^{'(t)}}_{\phi''_{t+1};{\boldsymbol\Phi}^{''(t)}} \equiv \left[\prod_{{\boldsymbol\Theta} \subseteq{\boldsymbol\Phi}^{(t)}} \tilde{K}
\left(\phi_{t+1}; {\boldsymbol\Theta}\right)\right]^{{\boldsymbol\Phi}^{'(t+1)}}_{{\boldsymbol\Phi}^{''(t+1)}},
\end{equation}
which is illustrated as shadowed parts in Fig.~\ref{fig:TensorRep}d and e for $t=2$.
To replace the summation form in Eq.~\eqref{eq:cumulantEx} with the tensor product form, suffices it to demonstrate
\begin{equation} \label{eq:Corr-Kernal-tensor}
  {\mathbb C}^{{\boldsymbol\Phi}^{'(t+1)}}_{{\boldsymbol\Phi}^{(t+1)}} =  {\mathbb K}^{\phi'_{t+1};{\boldsymbol\Phi}^{'(t)}}_{\phi_{t+1};{\boldsymbol\Phi}^{''(t)}} {\mathbb C}^{{\boldsymbol\Phi}^{''(t)}}_{{\boldsymbol\Phi}^{(t)}},
\end{equation}
or equivalently [with Eq.~\eqref{eq:def_C_updwn} and Eq.~\eqref{eq:cumulantEx} applied to the l.h.s. and Eq.~\eqref{eq:Corr-Kernal-tensor} and Eq.~\eqref{eq:CombinedKT-app} recursively applied to the r.h.s.],
\begin{equation} \label{eq:sum_C-prod_K}
{\mathbb P}^{{\boldsymbol\Phi}^{'(t)},{\boldsymbol\Phi}^{''(t)}}_{{\boldsymbol\Phi}^{(t)}} \sum_{{\boldsymbol\Theta}''\subseteq{\boldsymbol\Phi}^{''(t)}}\delta^{\mathbf 0}_{\bar{\boldsymbol\Theta}''}\sum_{\{{\boldsymbol\Theta}''_i\}}
%_{\{{\boldsymbol\Theta}_1,{\boldsymbol\Theta}_2,\ldots,{\boldsymbol\Theta}_n\}}
\tilde{C}_{{\boldsymbol\Theta}''_1}\tilde{C}_{{\boldsymbol\Theta}''_2}\cdots\tilde{C}_{{\boldsymbol\Theta}''_n} = \left[\prod_{{\boldsymbol\Theta} \subseteq{\boldsymbol\Phi}^{(t)}} \tilde{K}
\left({\boldsymbol\Theta}\right)\right]^{{\boldsymbol\Phi}^{'(t)}}_{{\boldsymbol\Phi}^{(t)}},
\end{equation}
where  $\{{\boldsymbol\Theta}''_1,{\boldsymbol\Theta}''_2, \ldots, {\boldsymbol\Theta}''_n\}$ are all possible partitions of the subset of indices ${\boldsymbol\Theta}''$ and $\bar{\boldsymbol\Theta}''\equiv {\boldsymbol\Phi}^{''(t)}\setminus{\boldsymbol\Theta}''$ is the complement. 
%To prove this equation, one just needs to compare the terms on both sides one by one.

\begin{figure}
    \centering
    \includegraphics[width=0.9\columnwidth]{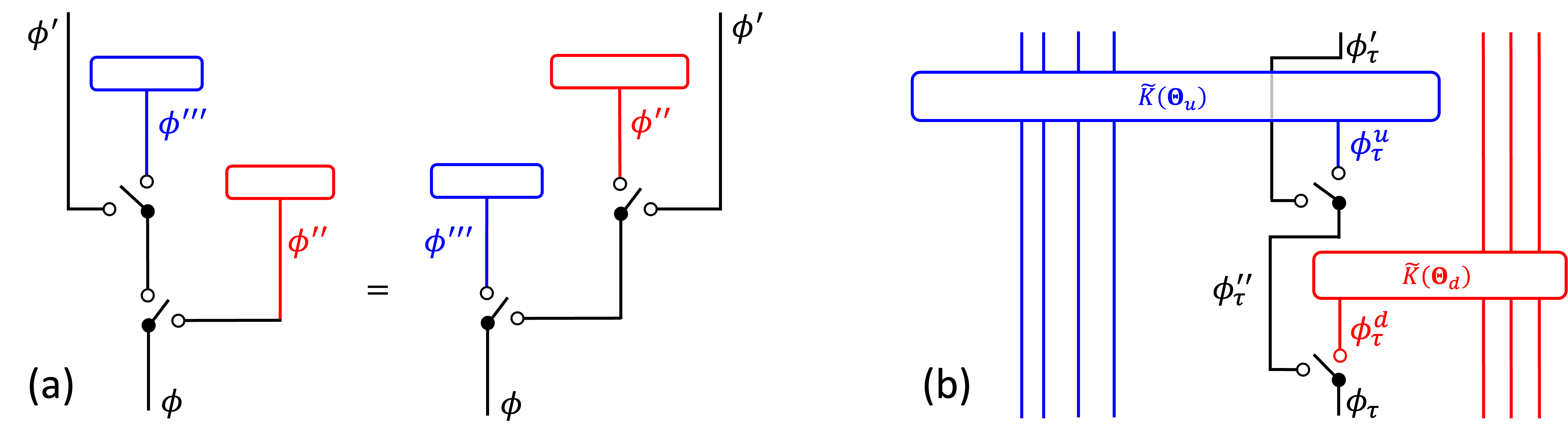}
    \caption{(a) Exchangeable product of picking tensors. (b) Tensor product of two kernel functions that have one common index $\phi_{\tau}$. Since each picking tensor sets one upward index to be zero and the other one to the same value of the downward index, $\phi^u_{\tau}$ and $\phi^d_{\tau}$ can not be simultaneously non-zero.}
    \label{fig-appendix:product-picking}
\end{figure}

Let us consider the tensor product of two terms in r.h.s. of Eq.~\eqref{eq:sum_C-prod_K}. Using Eq.~\eqref{eq.:def_kernel_combined}, we get 
 \begin{align} \label{eq:prod_KK_app}
   \left[ \tilde{K}\left({\boldsymbol\Theta}_u\right) \tilde{K}
\left({\boldsymbol\Theta}_d\right)\right]^{{\boldsymbol\Phi}^{'(t)}}_{{\boldsymbol\Phi}^{(t)}}
= & 
\delta^{{\boldsymbol\Phi}^{'(t)}}_{{\boldsymbol\Phi}^{(t)}}+{\delta}^{\bar{\boldsymbol\Theta}'_u}_{\bar{\boldsymbol\Theta}_u} {\delta}^{{\boldsymbol\Theta}'_u}_{\mathbf 0} \tilde{C}_{{\boldsymbol\Theta}_u}+{\delta}^{\bar{\boldsymbol\Theta}'_d}_{\bar{\boldsymbol\Theta}_d} {\delta}^{{\boldsymbol\Theta}'_d}_{\mathbf 0} \tilde{C}_{{\boldsymbol\Theta}_d}
%\nonumber \\ & 
+{\delta}^{\bar{\boldsymbol\Theta}'_u}_{\bar{\boldsymbol\Theta}''_u} {\delta}^{{\boldsymbol\Theta}'_u}_{\mathbf 0} \tilde{C}_{{\boldsymbol\Theta}''_u}{\delta}^{\bar{\boldsymbol\Theta}''_d}_{\bar{\boldsymbol\Theta}_d} {\delta}^{{\boldsymbol\Theta}''_d}_{\mathbf 0} 
 \tilde{C}_{{\boldsymbol\Theta}_d} .
 \end{align}
When the two cumulants $\tilde{C}_{{\boldsymbol\Theta}_u}$ and $\tilde{C}_{{\boldsymbol\Theta}_d}$ have at least one overlapping index, e.g., $\phi_{\tau}$ (see Fig.~\ref{fig-appendix:product-picking}b), the last term in Eq.~\eqref{eq:prod_KK_app} must vanish since $\delta^{\phi''_{\tau}}_{0}\tilde{C}_{{\boldsymbol\Theta}''_u}=0$. The physical meaning of this result is that noise at any time can affect the OQS only once through one cumulant (see Fig.~\ref{fig:TensorRep}c). 
When the two cumulants have no overlapping indices, i.e., ${\boldsymbol \Theta}_u\cap{\boldsymbol\Theta}_d=\varnothing$, the tensor product of the two kernel functions can be written as
\begin{align} \label{eq:prod_KK_app-II}
   \left[ \tilde{K}\left({\boldsymbol\Theta}_u\right) \tilde{K}
\left({\boldsymbol\Theta}_d\right)\right]^{{\boldsymbol\Phi}^{'(t)}}_{{\boldsymbol\Phi}^{(t)}}
= & 
\delta^{{\boldsymbol\Phi}^{'(t)}}_{{\boldsymbol\Phi}^{(t)}}+{\delta}^{\bar{\boldsymbol\Theta}'_u}_{\bar{\boldsymbol\Theta}_u} {\delta}^{{\boldsymbol\Theta}'_u}_{\mathbf 0} \tilde{C}_{{\boldsymbol\Theta}_u} +{\delta}^{\bar{\boldsymbol\Theta}'_d}_{\bar{\boldsymbol\Theta}_d} {\delta}^{{\boldsymbol\Theta}'_d}_{\mathbf 0} \tilde{C}_{{\boldsymbol\Theta}_d}
%\nonumber \\ & 
+{\delta}^{\bar{\boldsymbol\Theta}'_{u+d}}_{\bar{\boldsymbol\Theta}_{u+d}} {\delta}^{{\boldsymbol\Theta}'_{u+d}}_{\mathbf 0} \tilde{C}_{{\boldsymbol\Theta}_u}
 \tilde{C}_{{\boldsymbol\Theta}_d},
 \end{align}
 where ${\boldsymbol\Theta}_{u+d}\equiv {\boldsymbol\Theta}_{u} \bigcup {\boldsymbol\Theta}_{d}$ and $\bar{\boldsymbol\Theta}_{u+d}\equiv {\boldsymbol\Phi}^{(t)}\setminus{\boldsymbol\Theta}_{u+d}$. The tensor product with more kernel functions can be carried out similarly, and in the expansion, all products of cumulants with overlapping indices must vanish. Therefore, the r.h.s. of Eq.~\eqref{eq:sum_C-prod_K} can be written as
\begin{equation} \label{eq:rhs-prod-k}
%\left[\prod_{{\boldsymbol\Theta} \subseteq{\boldsymbol\Phi}^{(t)}} \tilde{K}
%\left({\boldsymbol\Theta}\right)\right]^{{\boldsymbol\Phi}^{'(t)}}_{{\boldsymbol\Phi}^{(t)}}=
\sum_{{\boldsymbol \Theta}\subseteq {\boldsymbol\Phi}^{(t)}}
{\delta}^{\bar{\boldsymbol\Theta}'}_{\bar{\boldsymbol\Theta}}\sum_{\{{\boldsymbol\Theta}_i\} } {\delta}^{{\boldsymbol\Theta}'}_{\mathbf 0} \tilde{C}_{{\boldsymbol\Theta}_1}\tilde{C}_{{\boldsymbol\Theta}_2}\cdots
 \tilde{C}_{{\boldsymbol\Theta}_k},
 \end{equation}
 where $\left\{{\boldsymbol\Theta}_1,{\boldsymbol\Theta}_2,\ldots,{\boldsymbol\Theta}_k\right\}$ is a non-overlapping partition of the indices ${\boldsymbol\Theta}$.
  The cumulant $\tilde{C}_{{\boldsymbol\Theta}_i}=0$ if any index $\phi\in {\boldsymbol \Theta}_i$ takes zero value, so
${\delta}^{{\boldsymbol\Theta}'_i}_{\mathbf 0} \tilde{C}_{{\boldsymbol\Theta}_i}= {\mathbb P}^{{\boldsymbol\Theta}'_i,{\boldsymbol\Theta}''_i}_{{\boldsymbol\Theta}_i}\tilde{C}_{{\boldsymbol\Theta}''_i}.$
Inserting this formula and ${\delta}^{\bar{\boldsymbol\Theta}'}_{\bar{\boldsymbol\Theta}}={\mathbb P}^{\bar{\boldsymbol\Theta}',\bar{\boldsymbol\Theta}''}_{\bar{\boldsymbol\Theta}}\delta^{\mathbf 0}_{\bar{\boldsymbol\Theta}''}$ into expression~\eqref{eq:rhs-prod-k}, we get the l.h.s. of Eq.~\eqref{eq:sum_C-prod_K}. Thus, the correlation tensor is transformed into a tensor product of kernel functions as in the r.h.s. of Eq.~\eqref{eq:sum_C-prod_K}.

\section{MPO representation of kernel tensors}\label{Append:MPO}
%\subsubsection{QR decomposition}
Here we derive the MPO form of combined kernel tensors $\mathbb{K}^{\phi'_{t+1};{\boldsymbol\Phi}^{'(t)}}_{\phi_{t+1};{\boldsymbol\Phi}^{(t)}}$. Each kernel tensor can be written as a matrix product state (MPS) through a series of QR decomposition~\cite{Schollwock2011}
\begin{align}
\tilde{K}_{\phi_{t+1};{\boldsymbol\Theta}} & =\sum_{\left\{{a}_i\right\}} \left[Q_{\phi_{t+1}}\right]_{{a}_m} \left[Q_{\phi_{\tau_m}}\right]_{{a}_{m},{a}_{m-1}} \cdots  \left[Q_{\phi_{\tau_2}}\right]_{{a}_{2}, {a}_{1}} \left[R_{\phi_{\tau_1}}\right]_{{a}_1}
\nonumber \\ &
\equiv Q_{\phi_{t+1}} Q_{\phi_{\tau_m}} \cdots  Q_{\phi_{\tau_2}}R_{\phi_{\tau_1}}
\nonumber \\
& %\equiv
\begin{tikzpicture}
    %first tensor
    \node at (-1,-0) {$\equiv$};
    \draw[thick] (0,0) --(1.0,0);
    \draw[thick] (0,-0.6) node[below]{$\phi_{t+1}$}--(0,0) node[fill=white]{$Q$};
    %%%%secdond tensor
    \node at (0.9, -0.3) {${a}_{m}$};
    \draw[thick] (1.0,0)--(2.4,0);
    \draw[thick] (1.7,-0.6) node[below]{$\phi_{\tau_m}$} --(1.7,0) node[fill=white]{$Q$};
    \node at (2.5, -0.3) {${a}_{m-1}$};
    %%%%%%%%%%%%%%%%%
    \node at (3.0,0) {$\Large{\ldots}$};
    %%%%third tensor
    \node at (3.7, -0.3) {${a}_{2}$};
    \draw[thick] (3.6,0)--(5.0,0);
    \draw[thick] (4.3,-0.6) node[below]{$\phi_{\tau_2}$} --(4.3,0) node[fill=white]{$Q$};
     %%%%last tensor
    \node at (5.3, -0.3) {${a}_{1}$};
    \draw[thick] (5.0,0)--(6.0,0);
    \draw[thick] (6.0, -0.6) node[below]{$\phi_{\tau_1}$} --(6.0,0) node[fill=white]{$R$};
    \node at (6.3, -0) {,};
\end{tikzpicture} 
\label{eq:QR-app}
\end{align}
where  ${\boldsymbol\Theta}\equiv\left\{\phi_{\tau_m}, \cdots, \phi_{\tau_1}\right\}$ and $\left\{{a}_i\right\}$ are left/right indices of the state tensors $Q$ and $R$.

We define the operator tensors $Q^{\phi'}_{\phi}\equiv P^{\phi',\phi''}_{\phi}Q_{\phi''}$, $R^{\phi'}_{\phi}\equiv P^{\phi',\phi''}_{\phi}R_{\phi''}$, and $\left[I^{\phi'}_{\phi}\right]_{{a},{b}}\equiv \delta^{\phi'}_{\phi}\delta_{{a},{b}}$, and obtain the MPO representation of  kernel tensors as
\begin{align}
\left[\tilde{K}\left(\phi_{t+1};{\boldsymbol\Theta}\right)\right]^{\phi'_{t+1};{\boldsymbol \Phi}^{'(t)}}_{\phi_{t+1}; {\boldsymbol \Phi}^{(t)}}
& = Q^{\phi'_{t+1}}_{\phi_{t+1}} I^{\phi'_{t}}_{\phi_{t}} \cdots Q^{\phi'_{\tau_m}}_{\phi_{\tau_m}} I^{\phi'_{\tau_m-1}}_{\phi_{\tau_m-1}}  \cdots  R^{\phi'_{\tau_1}}_{\phi_{\tau_1}} \delta^{\phi'_{\tau_1-1}}_{\phi_{\tau_1-1}}\cdots \delta^{\phi'_{1}}_{\phi_{1}}
\nonumber \\
& \begin{tikzpicture}
    %first Q tensor
    \node at (-1,0) {$\equiv$};
    \draw[thick] (0,0) --(0.7,0);
   \draw[thick] (0,0.6) node[above]{$\phi'_{t+1}$} --(0,0.0) node[fill=white]{$Q$} -- (0, -0.6) node[fill=white, below]{$\phi_{t+1}$};
    %%%%I tensor
    \draw[thick] (0.7,0)--(2.0,0);
    \draw[thick] (1.4,0.6) node[above]{$\phi'_{t}$} --(1.4, 0) node[fill=white]{$I$} -- (1.4, -0.6) node [below] {$\phi_{t}$};
    %%%%%%%%%%%%%%%%%
    \node at (2.4,0) {$\Large{\ldots}$};
    %%%%Q tensor
    \draw[thick] (2.8,0)--(4.1,0);
    \draw[thick] (3.4,0.6) node[above]{$\phi'_{\tau_{m}}$} --(3.4, 0) node[fill=white]{$Q$} -- (3.4, -0.6) node [below] {$\phi_{\tau_{m}}$};
    %%%%I tensor
    \draw[thick] (4.1,0)--(5.4,0);
    \draw[thick] (4.8,0.6) node[above]{$\phi'_{\tau_{m}-1}$} --(4.8, 0) node[fill=white]{$I$} -- (4.8, -0.6) node [below] {$\phi_{\tau_{m}-1}$};
    %%%%%%%%%%%%%%%%%
    \node at (5.8,0) {$\Large{\ldots}$};
    %%%R tensor
    \draw[thick] (6.2,0)--(6.8,0);
    \draw[thick] (6.8,0.6) node[above]{$\phi'_{\tau_{1}}$} --(6.8, 0) node[fill=white]{$R$} -- (6.8, -0.6) node [below] {$\phi_{\tau_{1}}$};
     %%%I tensors
     \draw[thick] (7.8,0.6) node[above]{$\phi'_{\tau_{1}-1}$} --(7.8, 0) -- (7.8, -0.6) node [below] {$\phi_{\tau_{1}-1}$};
     \node at (8.2,0) {$\Large{\ldots}$};
     \draw[thick] (8.8,0.6) node[above]{$\phi'_{1}$} --(8.8, 0) -- (8.8, -0.6) node [below] {$\phi_{1}$};
     \node at (9.1, 0) {,};
\end{tikzpicture}
\label{eq:kernelTesors-MPO-app}
\end{align}
where  the contraction of lateral indices $\{{a}_i\}$ between neighboring MPOs is implied.
The combined kernel tensor $\mathbb{K}^{\phi'_{t+1}; {\boldsymbol \Phi}^{'(t)}}_{\phi_{t+1};{\boldsymbol \Phi}^{(t)}}$ is then obtained by the contraction of all kernel tensors $\tilde{K}\left(\phi_{t+1};{\boldsymbol\Theta}\right)$ using Eq.~\eqref{eq:CombinedKT-app}. 
The MPO representation is preserved during the contraction of different kernel tensors through the following techniques
\begin{equation}
    \begin{tikzpicture}
    %%contraction of delta-delta tensors
    \draw[thick] (0,0.5)  --(0.0,0.1);
    \draw[thick] (0,-0.1) --(0.0,-0.5) ;
    \node at (0.5, 0.0) {$=$};
    \draw[thick] (1.0,0.4) --(1.0,-0.4);
    \node at(1.2,0) {,};
    %%contraction of R-delta tensors
    \node[fill=white] (I) at (2.3, 0.) {\small $I$};
    \node[fill=white] (R) at (3.2, 0.) {\small $R$};
    \draw[thick,blue] (I)  --+(0,0.6);
    \draw[thick,blue] (I)  --+(0,-0.6);
    \draw[thick,blue] (I)  --+(0.5,0);
    \draw[thick] (I)  -- (R);
    \draw[thick, black!20!green] (R) -- +(0.0, 0.6);
    \draw[thick, black!20!green] (R) -- +(0.0, -0.6);
    %%%%%%%%%%%%%%
    \node at (3.7, 0.0) {$=$};
    %%%%%%%%%%%%%
    \node [fill=white] (U1) at (4.3,0.0) {\small $U$};
    \node [fill=white] (L1) at (5.15, 0) {\small $\Lambda$};
    \node[fill=white] (V1) at (6.1,0.0) {\small $V^{\dagger}$};
    \draw[thick,blue] (U1) -- +(0.0,0.6);
    \draw[thick,blue] (U1) -- +(0.0,-0.6);
    \draw (U1)--(L1)--(V1);
    \draw[thick, black!20!green] (V1) -- +(0.0,0.6);
    \draw[thick, black!20!green] (V1) -- +(0.0,-0.6);
    \node at(6.6,0) {,};
     %%contraction of A-delta tensors\
    \node[fill=white](QI) at (8.3,0)  {\small $Q/I$};
    \node[fill=white](U) at (9.3, 0)  {\small $U$}; 
    \node[fill=white](L) at (10.1, 0) {\small $\Lambda$}; 
    \draw[thick,blue] (QI) -- +(0, 0.6);
    \draw[thick,blue] (QI) -- +(0, -0.6);
    \draw[thick,blue] (QI) -- +(-0.7, 0);
    \draw[thick,black!20!green] (U) -- +(0, 0.6);
    \draw[thick,black!20!green] (U) -- +(0, -0.6);
    \draw[thick] (QI) -- (U)--(L);
    \draw[thick,black!20!green] (L)-- +(0.5,0);
    %%%%%%%%%%%%%%
    \node at (10.9, 0.0) {$=$};
    %%%%%%%%%%%%%
    \node [fill=white] (U2) at  (11.8,0.0) {\small $U$};
    \node [fill=white] (L2) at (12.7,0.0)  {\small $\Lambda$};
    \node [fill=white] (V2) at (13.7,0.0)  {\small $V^{\dagger}$};
    \draw [thick,blue] (U2) -- +(-0.6,0);
    \draw [thick,blue] (U2) -- +(0,0.6);
    \draw [thick,blue] (U2) -- +(0,-0.6);
    \draw[thick] (U2)--(L2)--(V2);
    \draw [thick,black!20!green] (V2) -- +(0.6,0);
    \draw [thick,black!20!green] (V2) -- +(0,0.6);
    \draw [thick,black!20!green] (V2) -- +(0,-0.6);
    \node at(14.4,0) {,};
    \end{tikzpicture}
\end{equation}
\begin{equation}
     \begin{tikzpicture}
    \node[fill=white] (QI11) at (0.2, 0.5)  {\small $Q/I$};
    \node[fill=white] (U11) at (1.2, 0.5) {\small $U$};
    \node[fill=white] (L11) at (2., 0.5) {\small $\Lambda$};
    \node[fill=white] (QI12) at (0.2, -0.5)  {\small $Q/I$};
    \node[fill=white] (R1) at (1.2, -0.5)  {\small $R$};
    \draw[thick](QI11)--(U11)--(L11);
    \draw[thick](QI11)--(QI12)--(R1)--(U11);
    \draw[thick, blue] (QI11) --+(0.0,0.6);
    \draw[thick, blue] (QI11) --+(-0.7,0.0);
    \draw[thick, red] (QI12) --+(0.0,-0.6);
    \draw[thick, red] (QI12) --+(-0.7,0.0);
    \draw[thick, black!20!green] (U11) --+(0.0,0.6);
    \draw[thick,  black!20!green] (L11) --+(0.6,0.0);
    \draw[thick, black!20!green] (R1) --+(0.0,-0.6);
    %%contraction of A-delta tensors
    %%%%%%%%%%%
    \node at (2.9, 0.0) {$=$};
    %%%%%%%%%%%
    \node[fill=white](U12) at (3.8, 0.0) { \small $U$};
    \node[fill=white](L12) at (4.6, 0.0) { \small $\Lambda$};
    \node[fill=white](V12) at (5.5, 0.0) { \small $V^{\dagger}$};
    \draw[thick](U12)--(L12)--(V12);
    \draw [thick,blue] (U12)++(-0.2,0.05) -- +(-0.3,0);
    \draw [thick,red] (U12)++(-0.2,-0.05) -- +(-0.3,0);
    \draw [thick,blue] (U12) -- +(0,0.6);
    \draw [thick,red] (U12) -- +(0,-0.6);
    \draw [thick,black!20!green] (V12) -- +(0.6,0);
    \draw [thick,black!20!green] (V12) -- +(0,0.6);
    \draw [thick,black!20!green] (V12) -- +(0,-0.6);
    \node at (6.3, 0.0) {,};
     %%contraction of A-R tensors
    \node[fill=white](QI21) at (7.8, 0.5)  {\small $Q/I$};
    \node[fill=white](QI22) at (7.8, -0.5) {\small $Q/I$};
    \node[fill=white](U21)  at (8.8, 0.0)  {\small $U$};
    \node[fill=white] (L21) at (9.6, 0.0) {\small $\Lambda$};
    \draw[thick](QI21)--(U21)--(QI22)--(QI21);
    \draw[thick](U21)--(L21);
    \draw[thick, blue] (QI21) --+(0.0,0.6);
    \draw[thick, blue] (QI21) --+(-0.8,0.0);
    \draw[thick, red] (QI22) --+(0.0,-0.6);
    \draw[thick, red] (QI22) --+(-0.8,0.0);
    \draw[thick, black!20!green] (U21) --+(0.0,0.6);
    \draw[thick, black!20!green] (U21) --+(0.0,-0.6);
    \draw[thick,  black!20!green] (L21) --+(0.5,0.0);
    %%%%%%%%%%%
    \node at (10.4, 0.0) {$=$};
    %%%%%%%%%%%
    \node[fill=white](U22) at (11.3, 0.0) { \small $U$};
    \node[fill=white](L22) at (12.1, 0.0) { \small $\Lambda$};
    \node[fill=white](V22) at (13.1, 0.0) { \small $V^{\dagger}$};
    \draw[thick](U22)--(L22)--(V22);
    \draw [thick,blue] (U22)++(-0.2,0.05) -- +(-0.3,0);
    \draw [thick,red] (U22)++(-0.2,-0.05) -- +(-0.3,0);
    \draw [thick,blue] (U22) -- +(0,0.6);
    \draw [thick,red] (U22) -- +(0,-0.6);
    \draw [thick,black!20!green] (V22) -- +(0.6,0);
    \draw [thick,black!20!green] (V22) -- +(0,0.6);
    \draw [thick,black!20!green] (V22) -- +(0,-0.6);
    \node at (13.8, 0.0) {,};
    \end{tikzpicture}
\end{equation}
where the compact singular value decomposition $K_{{b}',{b}}=\sum_{{a=1}}^DU_{{b}',{a}} \lambda_{a} V^{\dagger}_{{a}, {b}}\equiv U_{{b}'}\Lambda \ V^{\dagger}_{{b}}$ is employed, the open arms in colors are indices kept intact (abbreviated as ${b}$ and ${b}'$) in the decomposition, $U$ and $V$ are semi-unitary, i.e., $\left(U^{\dag}U\right)_{{a}',{a}}= \delta_{{a}',{a}}$ and $\left(V^{\dag}V\right)_{{a},{a}'}= \delta_{{a},{a}'}$, and  $\lambda_{1}\ge\lambda_2\cdots \ge \lambda_{D}> 0$ are all of the non-zero singular values with the bond dimension $D\le\min\left(||{b}'||,||{b}||\right)$ being the rank of the tensor $K$. In the calculation, the singular values that are smaller than a given truncation precision $\xi$ are neglected to reduce the bond dimension (bond dimension $D$ set by {$\lambda_D/\lambda_1\ge \xi$}). The contraction of these kernel tensors is executed from  left to right, after which the result is restored to the QR representation as shown in Eq.~\eqref{eq:kernelTesors-MPO-app} by QR decomposition of $\left[ \tilde{K}(\bm{\Theta}_{u})\tilde{K}(\bm{\Theta}_{d})\right]^{\bm{0}}_{\bm{\Phi}^{''(t)}}$ and $\left[ \tilde{K}(\bm{\Theta}_{u})\tilde{K}(\bm{\Theta}_{d})\right]^{\bm{\Phi}^{'(t)}}_{\bm{\Phi}^{(t)}}\equiv {\mathbb P}^{{\boldsymbol\Phi}^{'(t)},{\boldsymbol\Phi}^{''(t)}}_{{\boldsymbol\Phi}^{(t)}} \left[ \tilde{K}(\bm{\Theta}_{u})\tilde{K}(\bm{\Theta}_{d})\right]^{\bm{0}}_{\bm{\Phi}^{''(t)}}$. Repeating the above contraction processes, we finally get
\begin{equation}
    \begin{tikzpicture}
        %vertical line
        \node at (-2,0) {${\mathbb K}^{\phi'_{t+1};{\boldsymbol\Phi}^{'(t)}}_{\phi''_{t+1};{\boldsymbol\Phi}^{''(t)}} =$};
        \draw[thick] (0.0,0.7) node[above]{\small $\phi'_{t+1}$} -- (0.0,-0.7) node[below]{\small $\phi_{t+1}$};
        \draw[thick] (2.0,0.7) node[above]{\small $\phi'_{t}$} -- (2.0,-0.7) node[below]{\small $\phi_{t}$};
        \draw[thick] (5.0,0.7) node[above]{\small $\phi'_{1}$} -- (5.0,-0.7) node[below]{\small $\phi_{1}$};
        %transverse line
        \draw[thick] (0.0, 0.0) node[fill=white]{\small $Q$} -- (2,0.0) node[fill=white]{\small $Q$}--(3.5,0.0) node[fill=white]{\small $\ldots$};
         \draw[thick] (4.0, 0.0)-- (5.0, 0.0) node[fill=white]{\small $R$};
         \node at (5.4,0.0){.};
    \end{tikzpicture}
\end{equation}

Note that in the QR decomposition in Eq.~\eqref{eq:QR-app} and hence in the MPO representation of kernel tensors in Eq.~\eqref{eq:kernelTesors-MPO-app}, the bond dimension of ${a}_{\tau_i}$ increases exponentially with $\min(i, t-i)$, making the QR decomposition of high-order tensors infeasible. Using Eq.~\eqref{eq:CombinedKT-app}, we can consider the contributions from different kernel tensors one by one and carry out the QR decomposition from right to left (the direction of time increasing) to construct the MPS representation of the EDMs. As guaranteed by the linearly increasing bond dimension theorem (Theorem~\ref{theorem:MPS}), the bond dimension of ${a}_i$ in such MPS construction of EDMs will increase at most linearly with $i$. But the number of cumulants and hence that of the kernel tensors could still increase exponentially with the evolution time. However, for many types of quantum baths, only a limited number of low-order kernel tensors are involved. For example, Gaussian baths have only first- and second-order cumulants,  $\tilde{C}_{\phi_{t+1}}$ and $\tilde{C}_{\phi_{t+1}, \phi_\tau}$. Moreover, very often quantum baths have finite memory times, which upper bounds the number of cumulants needed to construct the tensor networks of EDMs.

\section{Progressive construction of independent EDM vectors}
\label{Appendix:EDMvectors}

{Formally, one can construct the independent EDM vectors step by step with time evolution. 
First, for $t=1$, we choose $\left[r\right]^{{a}}_{\phi_1}$ such that
$${\mathbf R}^{{a}}=\left[r\right]^{{a}}_{\phi_1}{\boldsymbol\rho}^{\phi_1},$$
with ${a}=0,1,\ldots,N_{t=1}-1$.
For $t=2$, we set $\left[r\right]^{{a}}_{\phi_2,\phi_1}=\delta^0_{\phi_2}\left[r\right]^{{a}}_{\phi_1}$ for ${a}<N_{1}$ such that
$${\mathbf R}^{{a}<N_1}=\left[r\right]^{{a}}_{0,\phi_1}{\boldsymbol\rho}^{0,\phi_1}=\left[r\right]^{{a}}_{\phi_1}{\boldsymbol\rho}^{0,\phi_1}$$
and construct $N_2-N_1$ new vectors 
$${\mathbf R}^{N_1\le {a} < N_2}=\left[r\right]^{{a}}_{\phi_2\ne0,\phi_1}{\boldsymbol\rho}^{\phi_2\ne0,\phi_1}.$$
Continuing this process, for the EDMs at time $t+1$ we can set
$\left[r\right]^{{a}<N_t}_{\phi_{t+1},{\boldsymbol\Phi}^{(t)}}=\delta^0_{\phi_{t+1}}\left[r\right]^{{a}}_{{\boldsymbol \Phi}^{(t)}}$ such that
$${\mathbf R}^{{a}<N_t}=\left[r\right]^{{a}}_{0,{\boldsymbol \Phi}^{(t)}}{\boldsymbol \rho}^{0,{{\boldsymbol\Phi}^{(t)}}}=\left[r\right]^{{a}}_{{\boldsymbol \Phi}^{(t)}}{\boldsymbol \rho}^{0,{{\boldsymbol\Phi}^{(t)}}},$$
and choose $N_{t+1}-N_t$ new independent vectors 
$${\mathbf R}^{N_t\le {a}<N_{t+1}}=\left[r\right]^{{a}}_{\phi_{t+1}\ne 0,{\boldsymbol \Phi}^{(t)}}{\boldsymbol \rho}^{\phi_{t+1}\ne 0,{\boldsymbol\Phi}^{(t)}}.$$
Thus we can formally construct the independent vectors step by step of the time evolution through the procedure of choosing coefficients
$$\left[r\right]^{{a}<N_1}_{\phi_1} \rightarrow \left[r\right]^{{a}<N_1}_{0,\phi_1},\left[r\right]^{N_1\le{a}<N_2}_{\phi_2\ne 0,\phi_1}
\rightarrow \left[r\right]^{{a}<N_1}_{0,0,\phi_1},\left[r\right]^{N_1\le{a}<N_2}_{0,\phi_2\ne 0,\phi_1},
\left[r\right]^{N_2\le{a}<N_3}_{\phi_3\ne 0,\phi_2,\phi_1}\rightarrow \cdots \rightarrow 
\left[r\right]^{{a}<N_T}_{{\boldsymbol\Phi}^{(T)}}.$$
Without loss of generality, we choose
$${\mathbf R}^0={\boldsymbol \rho}^{{\boldsymbol\Phi}^{(t)}={\mathbf 0}}.$$
The EDMs can be obtained from the independent vectors through
$$\rho^{{\boldsymbol\Phi}^{(t)}}=\left[\bar{r}\right]^{{\boldsymbol\Phi}^{(t)}}_{{a}} R^{{a}}_t,$$
for any $t\le T$, where the inverse coefficient matrices satisfy
$$\left[r\right]_{{\boldsymbol\Phi}^{(t)}}^{{a}}
\left[\bar{r}\right]^{{\boldsymbol\Phi}^{(t)}}_{{b}}=\delta^{{a}}_{{b}}.$$

\section{Bond dimension for baths with finite memory times}\label{sec:FBDT-FMT-app}

Here we prove Corollary~\ref{Cor_MPS_memory}. The memory time $T_m$ of a bath is defined by the condition that  the cumulants $\tilde{C}_{\phi_{t_l},\cdots,\phi_{t_1}}$ of the bath vanishes if $t_l-t_1>T_m$. From  Eq.~\eqref{eq:CombinedKT-app}, one can get that the combined kernel tensors satisfy 
\begin{equation}
\mathbb{K}^{{\phi_{t+1};{\boldsymbol \Phi}^{(t)}}}_{\phi'_{t+1};{\boldsymbol \Phi}'^{(t)}}=\mathbb{K}^{{\phi_{t+1}; {\boldsymbol \Phi}^{(t:t-T_m+1)}}}_{\phi'_{t+1}; {\boldsymbol \Phi}'^{(t:t-T_m+1)}}\delta^{{\boldsymbol \Phi}^{(t-T_{m})}}_{{\boldsymbol \Phi}'^{(t-T_{m})}},
\end{equation}
for $t>T_m$, where ${\boldsymbol \Phi}^{(t:t-T_m+1)}\equiv \left\{\phi_t, \phi_{t-1},\ldots,\phi_{t-T_m+1}\right\}$. 
Therefore, the set of EMDs $\left\{ \rho^{{\boldsymbol \Phi}^{(t:t-T_m+1)},{\mathbf 0}}\equiv \rho^{{\boldsymbol \Phi}^{(t)}} \delta^{\mathbf{0}}_{{\boldsymbol \Phi}^{(t-T_m)}}\right\}$ are sufficient for solving $\rho(t')$ for $t'>t$. All the EDM vectors ${\boldsymbol \rho}^{{\boldsymbol \Phi}^{(t:t-T_m+1)},{\mathbf 0}}$
form a linear space with dimension $\le N_{T_m}$ and hence can be expanded by at most $N_{T_m}$ independent vectors ${\mathbf R}^{{a}}={\left[{r}\right]}_{{\boldsymbol \Phi}^{(t:t-T_m+1)}}^{{a}}{\boldsymbol \rho}^{{\boldsymbol \Phi}^{(t:t-T_m+1)},{\mathbf 0}}$ with ${a}=0, 1, \ldots, N_{T_m}-1$, that is,
\begin{equation}
    {\boldsymbol \rho}^{{\boldsymbol \Phi}^{(t:t-T_m+1)},{\mathbf 0}}={\left[\bar{r}\right]}^{{\boldsymbol \Phi}^{(t:t-T_m+1)}}_{{a}} {\mathbf R}^{{a}}.
\end{equation}
The procedure for obtaining the MPS representation of the EDM tensor is illustrated in Fig.~\ref{fig:finitMT-app} by taking $T_m = 2$ as an example. It follows the same SVD procedure as in Fig.~\ref{fig:Theorem_MPS}c, except that the isometries $\bar{r}^{\bm{\Phi}^{(t)}}_{a}$ and $ r^{a}_{\bm{\Phi}^{(t)}}$ stand for $\bar{r}^{\bm{\Phi}^{(t:t-T_m+1)}, \bm{0}}_{a}$ and $ r^{a}_{\bm{\Phi}^{(t:t-T_m+1)},\bm{0}}$, respectively, due to the finite memory time. 
The decomposition process shown in Fig.~\ref{fig:finitMT-app} is repeated until $t-T_m=0$, after which the same procedure as in Fig.~\ref{fig:Theorem_MPS}c is followed and the bond dimension is upper bounded by $N_t$. Consequently, the bond dimension between local tensors defined in Eq.~\eqref{eq:MPS_DM} and Eq.~\eqref{eq:MPS_EDM} is upper-bounded by $D_{\tau}\le \min(N_{\tau},N_{T_m})$.

\begin{figure}
    \centering
    \includegraphics[width=0.95\linewidth]{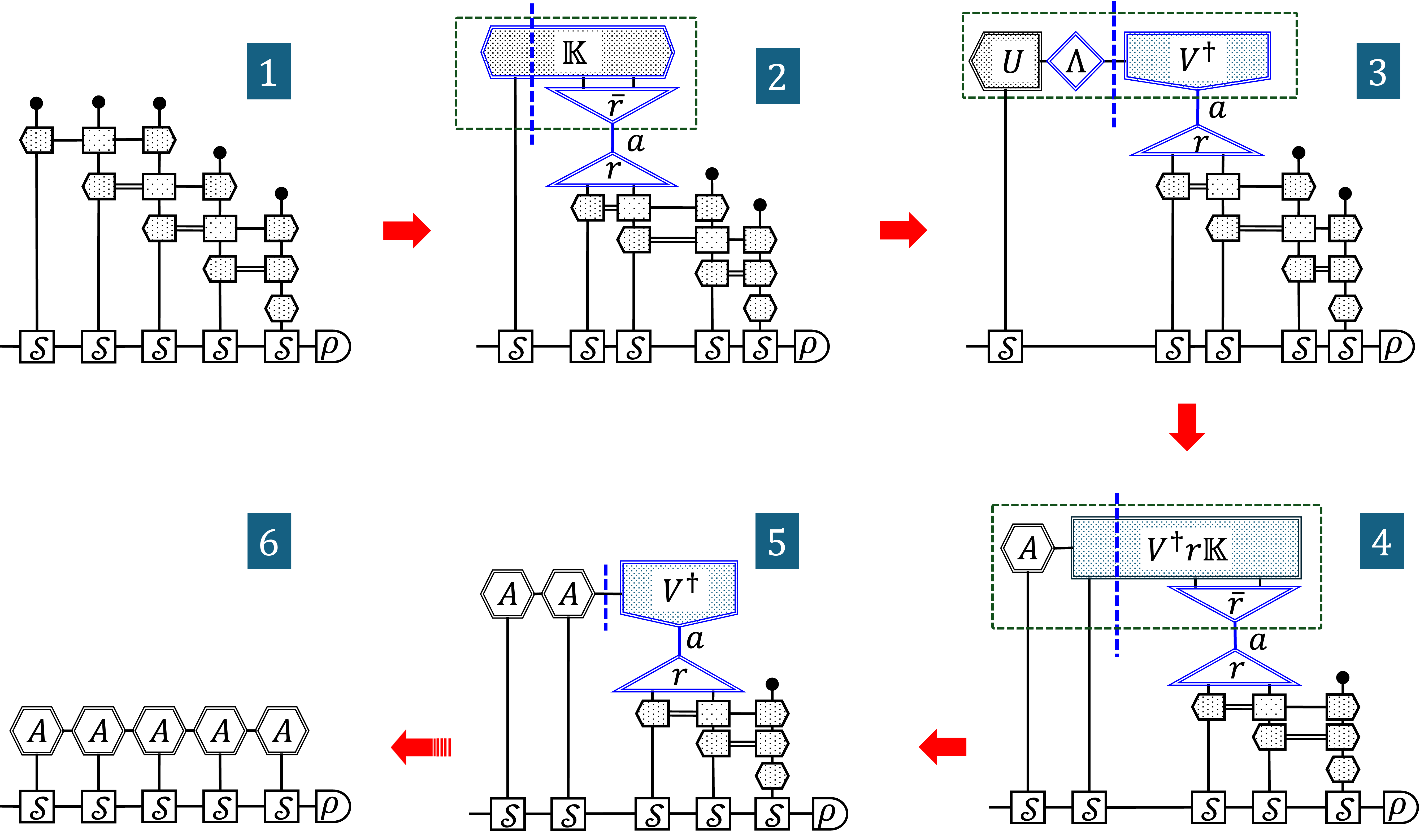}
    \caption{Tensor-network diagrams showing how the MPS representation of the density matrix is obtained for an OQS in a quantum bath with finite memory time ($T_m=2$ an an example). The bond dimension between local tensors is upper-bounded by the number of independent EDM vectors in the finite memory time (the range of index $a$ of the isometry in diagram 2).}
    \label{fig:finitMT-app}
\end{figure}

\section{Spin-boson model}\label{sec:SBM-app}
The bath is Gaussian. At zero temperature, the noise ${B}(t)=\sum_{k} g_k \left({a}^{\dagger}_ke^{i\omega_k t}+{\rm h.c.}\right)$ has only the second-order cumulants being non-zero, which are
\begin{equation}
	\tilde{C}_{\phi_{\tau_2}=2,\phi_{\tau_1}=2}=\text{Re}[f(\tau_2-\tau_1)], ~~
	\tilde{C}_{\phi_{\tau_2}=2,\phi_{\tau_1}=1}=2\text{Im}[f(\tau_2-\tau_1)],
\end{equation}
where $f(\tau)\equiv \int d\omega J(\omega) e^{-i \omega \tau}$, with $\mathcal{B}_{\phi_t=1}\equiv\mathcal{B}^{-}(t)$ and $\mathcal{B}_{\phi_t=2}\equiv\mathcal{B}^{+}(t)$.
The  kernel tensor has QR decomposition $\tilde{K}_{\phi_{t+1},\phi_{\tau+1}}=\left[Q_{\phi_{t+1}}\right]_{{a}} \left[R_{\phi_{\tau+1}}\right]_a$ with lateral bond dimension $D_a=2$ and the combined kernel tensor 
\begin{equation}
\mathbb{K}^{\phi'_{t+1};\boldsymbol{\Phi}^{'(t)}}_{\phi_{t+1}; \boldsymbol{\Phi}^{(t)}} =Q^{\phi'_{t+1}}_{\phi_{t+1}} Q^{\phi'_{t}}_{\phi_{t}} \cdots R^{\phi'_{1}}_{\phi_{1}}, 
\end{equation} 
where $Q^{\phi'}_{\phi} =P^{\phi',\phi''}_{\phi}Q_{\phi''}$ and $R^{\phi'}_{\phi}=P^{\phi',\phi''}_{\phi}R_{\phi''}$ with
\[Q_{\phi_{t+1}=0}=(1,0),~~ Q_{\phi_{t+1}=1}=(0,0), ~~Q_{\phi_{t+1}=2}=(0,1),\]
\[R_{\phi_{1}=0}=\left(\begin{array}{cc}
     1 &  \\
     0 & 
\end{array}\right),~~ R_{\phi_{1}=1}=\left(\begin{array}{cc}
     0  & \\
     2 \text{Im}[f(t)] &  
\end{array}\right), ~~ ~~ R_{\phi_{1}=2}=\left(\begin{array}{cc}
     0 &  \\
      \text{Re}[f(t)] & 
\end{array}\right),\]
and
\[Q_{\phi_{\tau+1}=0}=\left(\begin{array}{cc}
     1 & 0  \\
     0 & 1
\end{array}\right),~~ Q_{\phi_{\tau+1}=1}=\left(\begin{array}{cc}
     0  & 0\\
     2 \text{Im}[f(t-\tau)] & 0 
\end{array}\right), ~~ ~~ Q_{\phi_{\tau+1}=2}=\left(\begin{array}{cc}
     0 &  0 \\
      \text{Re}[f(t-\tau)] & 0
\end{array}\right),\]
for $0<\tau<t$.

\begin{figure}
    \centering
    \includegraphics[width=0.5\linewidth]{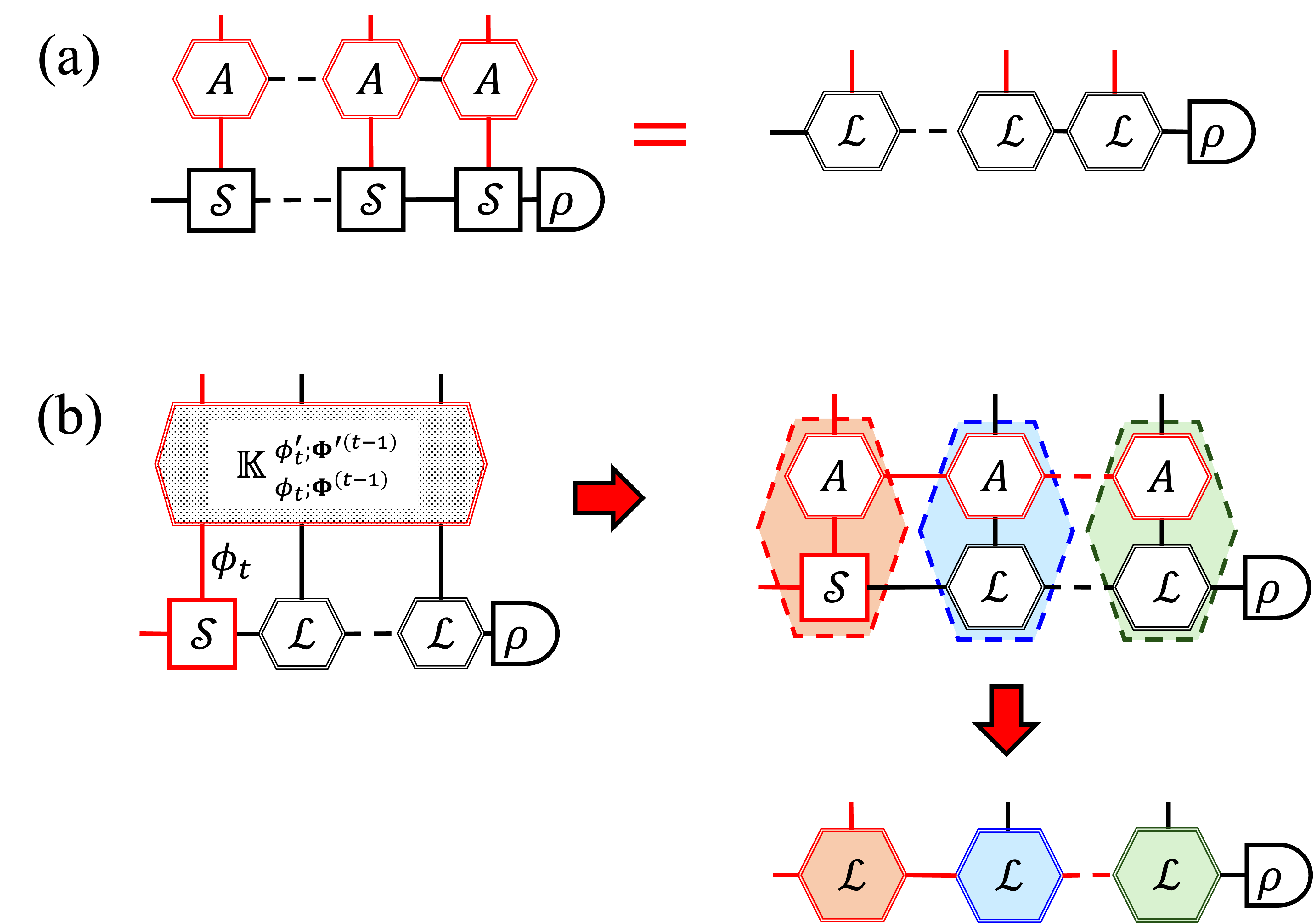}
    \caption{(a) Super-operator-valued MPS representation of an EDM. (b) Algorithm for recursive construction of the MPS representation of an EDM.}
    \label{fig:MPSAlgorithm-app}
\end{figure}

{
In programming, we introduce super-operator-valued local tensors ${\mathcal L}^{\phi_{\tau}}=A^{\phi_{\tau}}_{\phi'_{\tau}}{\mathcal S}^{\phi'_{\tau}}$ for MPS representation of the EDMs shown in Eq.~\eqref{eq:MPS_EDM}, which results in
\begin{equation}
    \rho^{{\bm \Phi}^{(t)}}=[\mathcal{L}^{\phi_{t}}]_{a_{t-1}}[\mathcal{L}^{\phi_{t-1}}]_{a_{t-1},a_{t-2}} \cdots [\mathcal{L}^{\phi_{2}}]_{a_{2},a_{1}}[\mathcal{L}^{\phi_{1}}]_{a_{1}} \rho(0).
\end{equation}
The diagram illustration is shown in Fig.~\ref{fig:MPSAlgorithm-app}a. }
The standard MPO-MPS algorithms are used to contract the combined kernel tensor and the EDM tensor to evaluate the EDMs~\cite{Schollwock2011}. {Such a process is shown in Fig.~\ref{fig:MPSAlgorithm-app}b.} With the contraction carried out from the left boundary to the right boundary, for example,  the following SVD technique is used
\begin{equation}
    \begin{tikzpicture}
         \node[fill=white] (Q)  at (2.7,1.0) {\small $Q$};
         \node[fill=white] (L1) at (0.7,0.5) {\small $\Lambda$};
         \node[fill=white] (V1) at (1.7,0.5)   {\small $V^{\dagger}$};
         \node[fill=white] (L2) at  (2.7,0.0) {\small $\mathcal{L}$};
         \draw[thick](L1)--(V1)--(L2)--(Q);
         \draw[thick](V1)--(Q);
         \draw[thick, black!20!green] (Q) --+(0.0,0.6);
         \draw[thick, black!20!green] (Q) --+(0.6,0.0);
         \draw[thick, blue] (V1) --+(0.0,0.9);
         \draw[thick, blue] (L1) --+(-0.6,0.0);
         \draw[thick, red] (L2) --+(0.6,0.0);
        %%%%%%%%%%%%%%%%%%%%%%
        \node at (4.0, 0.5) {$\approx$};
        %%%%%%%%%%%%%%%%%%%%%%
         \node[fill=white] (U)  at (5.2,0.5) {\small $U$};
         \node[fill=white] (L3) at (6.2,0.5)   {\small $\Lambda$};
         \node[fill=white] (V2) at  (7.2,0.5) {\small $V^{\dagger}$};
         \draw[thick] (U)--(L3)--(V2);
         \draw[thick, blue] (U) --+(0.0,0.6);
         \draw[thick, blue] (U) --+(-0.7,0.0);
         \draw[thick, black!20!green] (V2) --+(0.0,0.6);
         \draw[thick, black!20!green] (V2)++(0.25,0.05) --+(0.4,0.0);
         \draw[thick, red] (V2)++(0.25,-0.05) --+(0.4,0.0);
         \node at (8.0,0.5) {,};
    \end{tikzpicture}
\end{equation}
where in the approximation the singular values that are too small ($\lambda_a <\xi \lambda_{d^2+1}$) are dropped {with $\xi$ being the} truncation precision.

\begin{figure}
    \centering
    \includegraphics[width=0.9\columnwidth]{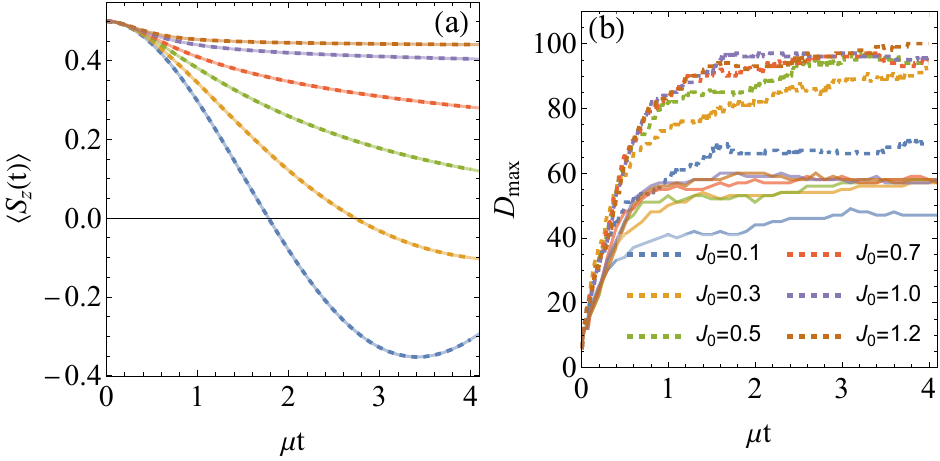}
    \caption{The evolution of the central spin polarization $\left\langle S_z(t)\right\rangle$ (a) and the maximal bond dimension (b) for truncation precisions $\xi=10^{-5}$ (thick lines) and $\xi=10^{-6}$ (dashed lines).  The parameters are the same as in Fig.~\ref{fig:SBM}.}
    \label{fig:truncation}
\end{figure}

In the numerical simulation, to reduce the error due to the discretization of time, we expand a small step of evolution to the second order of time step $\epsilon$,
\begin{align} \label{eq:high-order-app}
    e^{\epsilon \mathcal{H}^{-}(t)}=& \mathcal{I}+\epsilon \mathcal{H}^{-}(t)+\frac{\epsilon^2}{2}\left[\mathcal{H}^{-}(t)\right]^{2}+O\left(\epsilon^{3}\right) \nonumber \\
    =&\left[\mathcal{I}+\frac{1+i}{2} \epsilon \mathcal{H}^{-}(t)\right]\left[\mathcal{I}+\frac{1-i}{2} \epsilon \mathcal{H}^{-}(t)\right]  +O\left(\epsilon^{3}\right),
\end{align} 
which can be written as the tensor product form
\begin{equation}\label{eq:2ndorderexpasion}
    e^{\epsilon \mathcal{H}^{-}(t)}=\mathcal{B}_{\phi_t} \mathcal{B}_{\psi_t}  \mathcal{S}^{\phi_t}_2 \mathcal{S}^{\psi_t}_{1}  +O\left(\epsilon^3\right),
\end{equation}
where $\mathcal{S}_{1}^{\phi_t=0}\equiv \mathcal{S}_{2}^{\phi_t=0}\equiv \mathcal{I}_{\rm S}$, $\mathcal{S}^{\phi_t \neq 0}_{1}\equiv \frac{1-i}{2} \mathcal{S}^{\phi_t}$, and $\mathcal{S}^{\phi_t \neq 0}_{2}\equiv \frac{1+i}{2}\mathcal{S}^{\phi_t}$.
Using this formula, we generalize the tensor network algorithm developed in the main text to the second order of $\epsilon$, with each term  $\mathcal{B}_{\phi_t}\mathcal{S}^{\phi_t}$ replaced with $\mathcal{B}_{\phi_t} \mathcal{B}_{\psi_t}  \mathcal{S}^{\phi_t}_2 \mathcal{S}^{\psi_t}_{1}$. It is straightforward to generalize the second-order expansion of $\epsilon$ to higher orders.

Figure~\ref{fig:truncation} compares the  calculated spin dynamics and the maximal bond dimension $D_{\rm \max}$ needed to evaluate $\rho^{\bm{\Phi}^{(t)}}$ for truncation precision $\xi=10^{-5}$ (the same as in Fig.~\ref{fig:SBM} in the main text) and $\xi=10^{-6}$.  The calculation of the spin dynamics for truncation precision $\xi=10^{-5}$ is sufficiently accurate. As expected, smaller truncation precision induces a faster increase of the bond dimension. However, for all coupling strengths and both truncation precisions, the increase is not faster than linear.

\section{Gaudin model}\label{sec:GM-app}

\subsection{MPS representation of EDMs}

As a feature of non-interacting spin-1/2's, the correlation tensor for the $k$-th bath spin at infinite temperature can be decomposed into a product of local tensors as}
 \begin{equation}
 \mathbb{C}_{k; {\boldsymbol \Phi}^{(T)}}=\left[{ A}_{k;\phi_{T}}\right]_{0,a_{T-1}}  \cdots\left[{A}_{k;\phi_{2}}\right]_{a_{2},a_{1}} \left[{A}_{k;\phi_1}\right]_{a_1,0},  
 \end{equation}
 where $\left[{A}_{k;\phi_t}\right]_{a_t,  a_{t-1}}\equiv\frac{1}{2}\text{Tr}\left[\sigma_{a_t} \mathcal{B}_{k;\phi_{t}} \sigma_{a_{t-1}}\right]$ with $B_{k;1/2/3}\equiv g_k{J}_{k;x/y/z}$, $\sigma_{0}\equiv I$ and $\sigma_{1/2/3}\equiv 2 J_{k;x/y/z}$ with the lateral  bond dimension $D_a=4$ for $a_{\tau}$. 
 The EDM of the central spin coupled to the first spin is
 $$
 \rho^{\bm{\Phi}^{(T)}}_{L=1} = {\mathbb P}^{{\boldsymbol \Phi}^{(T)},{\boldsymbol \Phi}^{'(T)}}_{{\boldsymbol \Phi}^{''(T)}}\mathbb{C}_{1; {\boldsymbol \Phi}^{'(T)}}{\mathcal S}^{{\boldsymbol \Phi}^{''(T)}}\rho(0),
$$
 whose MPS can be obtained with the SVD described in Appendix~\ref{sec:SBM-app}; then the MPS representation of the EDM coupled to the first $L+1$ bath spins can be recursively obtained using Eq.~\eqref{eq:Gaudin-recursive}.
 % $$
 %\rho^{\bm{\Phi}^{(T)}}_{L+1} = {\mathbb P}^{{\boldsymbol \Phi}^{(T)},{\boldsymbol \Phi}^{'(T)}}_{{\boldsymbol \Phi}^{''(T)}}\mathbb{C}_{L+1; {\boldsymbol \Phi}^{'(T)}}{\rho}^{{\boldsymbol \Phi}^{''(T)}}_L.
 %$$

\subsection{Central spin dynamics}
The dynamics of the central spin at any given time $t$ can be obtained from
\begin{align}\label{eq:spinDynamicsApp}
        \left\langle S_{\alpha}(t)\right\rangle=& \epsilon^{-1}\text{Tr}\left[\rho^{{\boldsymbol \Phi}^{(T)}}  \delta_{\phi_{t}}^{2 \alpha-1} \delta^{{\mathbf 0}}_{{\boldsymbol \Phi}^{(T)}/\{\phi_{t}\}} \right] +O\left(\epsilon\right),
\end{align}
since
\begin{align}
\epsilon \left\langle S_{\alpha}(t)\right\rangle = & 
\text{Tr}\left[\epsilon\mathcal{S}_{\alpha}^+\rho(t)\right]\equiv \text{Tr}\left[\mathcal{S}^{2{\alpha}-1}\rho(t)\right]
\nonumber \\
=& \text{Tr}\left[\mathcal{S}^{2{\alpha}-1}e^{\epsilon \mathcal{H}^{-}(t)}\cdots e^{\epsilon \mathcal{H}^{-}(1)} \chi(0) \right] \nonumber\\
=& \text{Tr}\left[e^{\epsilon\mathcal{H}^{-}(T)} \cdots e^{\epsilon\mathcal{H}^{-}(t+1)}  \mathcal{S}^{2{\alpha}-1}e^{\epsilon \mathcal{H}^{-}(t)}\cdots e^{\epsilon \mathcal{H}^{-}(1)} \chi(0)\right] \nonumber \\
= & 
\text{Tr}\left[\rho^{{\boldsymbol \Phi}^{(T)}}  \delta_{\phi_{t}}^{2 \alpha-1} \delta^{{\mathbf 0}}_{{\boldsymbol \Phi}^{(T)}/\{\phi_{t}\}} \right] +O\left(\epsilon^2\right),\nonumber
\end{align}
where ${\boldsymbol \Phi}^{(T)}/\{\phi_{t}\}$ is the complement of $\{\phi_{t}\}$ and Eq.~\eqref{eq:EDM-phy} has been used in the last step.

In numerical calculation, we expand each step of evolution to the second order using Eq.~\eqref{eq:high-order-app}. For this second-order expansion, the EDM tensor is generalized to $$\rho^{\bm{\Phi}^{(T)},\bm{\Psi}^{(T)}} \equiv \mathbb{P}^{\bm{\Phi}^{(T)},\bm{\Phi}^{'(T)}}_{\bm{\Phi}^{''(T)}} \mathbb{P}^{\bm{\Psi}^{(T)},\bm{\Psi}^{'(T)}}_{\bm{\Psi}^{''(T)}} \mathbb{C}_{\bm{\Phi}^{'(T)},\bm{\Psi}^{'(T)}} \mathcal{T}\left[\mathbb{S}^{\bm{\Phi}^{''(T)}}_{2}\mathbb{S}^{\bm{\Psi}^{''(T)}}_1\right] \rho(0), $$
where $\mathbb{C}_{\bm{\Phi}^{(T)},\bm{\Psi}^{(T)}}\equiv \text{Tr}_B\left[\mathcal{B}_{\phi_T}(T) \mathcal{B}_{\psi_T}(T) \cdots \mathcal{B}_{\phi_1}(1)\mathcal{B}_{\psi_1}(1) \Omega\right]$ and $\mathbb{S}^{\bm{\Phi}^{(T)}}_{i}\equiv \mathcal{S}^{\phi_{T}}_{i}(T) \cdots \mathcal{S}^{\phi_{1}}_{i}(1)$ for $i=1$ or $2$.
Then, the evaluation of spin dynamics 
in Eq.~\eqref{eq:spinDynamicsApp} is modified to be
\begin{equation}\label{eq:Spin_2nd_Gaudin}
     \left\langle S_{\alpha}(t)\right\rangle=\frac{1+i}{\epsilon}\text{Tr}\left[\rho^{{\boldsymbol \Phi}^{(T)},{\boldsymbol \Psi}^{(T)}}  \delta^{{\mathbf 0}}_{{\boldsymbol \Phi}^{(T)}} \delta_{\psi_{t+1}}^{2 \alpha-1} \delta^{{\mathbf 0}}_{{\boldsymbol \Psi}^{(T)}/\{\psi_{t+1}\}}   \right]+O\left(\epsilon^2\right).
\end{equation}
To obtain this expression, we expand small steps of evolution $e^{\epsilon {\mathcal H}^-(\tau)}$ to the second order using Eq.~\eqref{eq:2ndorderexpasion} and get
\begin{align}\label{eq:SpinUnder2ndExpan}
  \frac{\epsilon}{1+i}\left\langle S_{\alpha}(t)\right\rangle 
    = & \text{Tr}_{\rm B}\left[\mathcal{S}^{2{\alpha}-1}_1 \rho(t)\right]
    \nonumber \\ = &  \text{Tr}\Big[ \mathcal{S}^{2{\alpha}-1}_1 
     {\mathcal B}_{\phi_{t}}{\mathcal B}_{\psi_{t}} {\mathcal S}^{\phi_{t}}_2{\mathcal S}^{\psi_{t}}_1\cdots {\mathcal B}_{\phi_1}{\mathcal B}_{\psi_1} {\mathcal S}^{\phi_1}_2{\mathcal S}^{\psi_1}_1\chi(0) \Big]  +O\left(\epsilon^3\right)
     \nonumber \\
          = & \text{Tr}\Big[{\mathcal B}_{\phi_{T}}{\mathcal B}_{\psi_{T}} {\mathcal S}^{\phi_{T}}_2{\mathcal S}^{\psi_{T}}_1\cdots {\mathcal B}_{\phi_{t+2}}{\mathcal B}_{\psi_{t+2}} {\mathcal S}^{\phi_{t+2}}_2{\mathcal S}^{\psi_{t+2}}_1 {\mathcal B}_{\phi_{t+1}}{\mathcal S}_2^{\phi_{t+1}} 
    \nonumber \\ &  \times  \mathcal{S}^{2{\alpha}-1}_1 
     {\mathcal B}_{\phi_{t}}{\mathcal B}_{\psi_{t}} {\mathcal S}^{\phi_{t}}_2{\mathcal S}^{\psi_{t}}_1\cdots {\mathcal B}_{\phi_1}{\mathcal B}_{\psi_1} {\mathcal S}^{\phi_1}_2{\mathcal S}^{\psi_1}_1\chi(0) \Big]  +O\left(\epsilon^3\right),
\end{align}
where the last equation above is obtained by inserting tensor products like ${\mathcal B}_{\phi} {\mathcal S}_i^{\phi}$ in the beginning of the expression within the trace operation since trace of commutators always vanishes. Using the substitution 
$
\mathcal{S}^{2{\alpha}-1}_1=P^{\psi_{t+1},2{\alpha}-1}_{\psi'_{t+1}}{\mathcal B}_{\psi_{t+1}}\mathcal{S}^{\psi'_{t+1}}_1
$
for $2\alpha-1>0$, we transfer Eq.~\eqref{eq:SpinUnder2ndExpan} to
\begin{align}
  \frac{\epsilon}{1+i}\left\langle S_{\alpha}(t)\right\rangle 
       = & {\mathbb C}_{{\boldsymbol \Phi}^{(T)},{\boldsymbol \Psi}^{(T)}}
    \text{Tr}\left[\hat{\mathcal T}{\mathbb S}_2^{{\boldsymbol \Phi}^{(T)}}{\mathcal S}_1^{\psi_T}\cdots{\mathcal S}_1^{\psi_{t+2}}{\mathcal S}_1^{\psi'_{t+1}}{\mathcal S}_1^{\psi_t}\cdots{\mathcal S}_1^{\psi_1}\rho(0) \right] P^{\psi_{t+1},2\alpha-1}_{\psi'_{t+1}}
    \nonumber  \\     =  & \text{Tr}\left[\rho^{{\boldsymbol \Phi}^{(T)},{\boldsymbol \Psi}^{(T)}}  \delta^{{\mathbf 0}}_{{\boldsymbol \Phi}^{(T)}} \delta_{\psi_{t+1}}^{2 \alpha-1} \delta^{{\mathbf 0}}_{{\boldsymbol \Psi}^{(T)}/\{\psi_{t+1}\}}   \right] +O\left(\epsilon^3\right),
\end{align}
which leads to Eq.~\eqref{eq:Spin_2nd_Gaudin}.
Note that in the model considered here, we have $S_{\alpha}(t)=S_{\alpha}(t+1)$. In general, there could be evolution of the system operators due to dynamical control. In such cases, one just need to modify the formula by writing $S_{\alpha}(t)$ as a linear combination of $S_{\beta}(t+1)$.

 \subsection{Error induced by bond dimension cutoff}

 \begin{figure}
     \centering
     \includegraphics[width=0.9\linewidth]{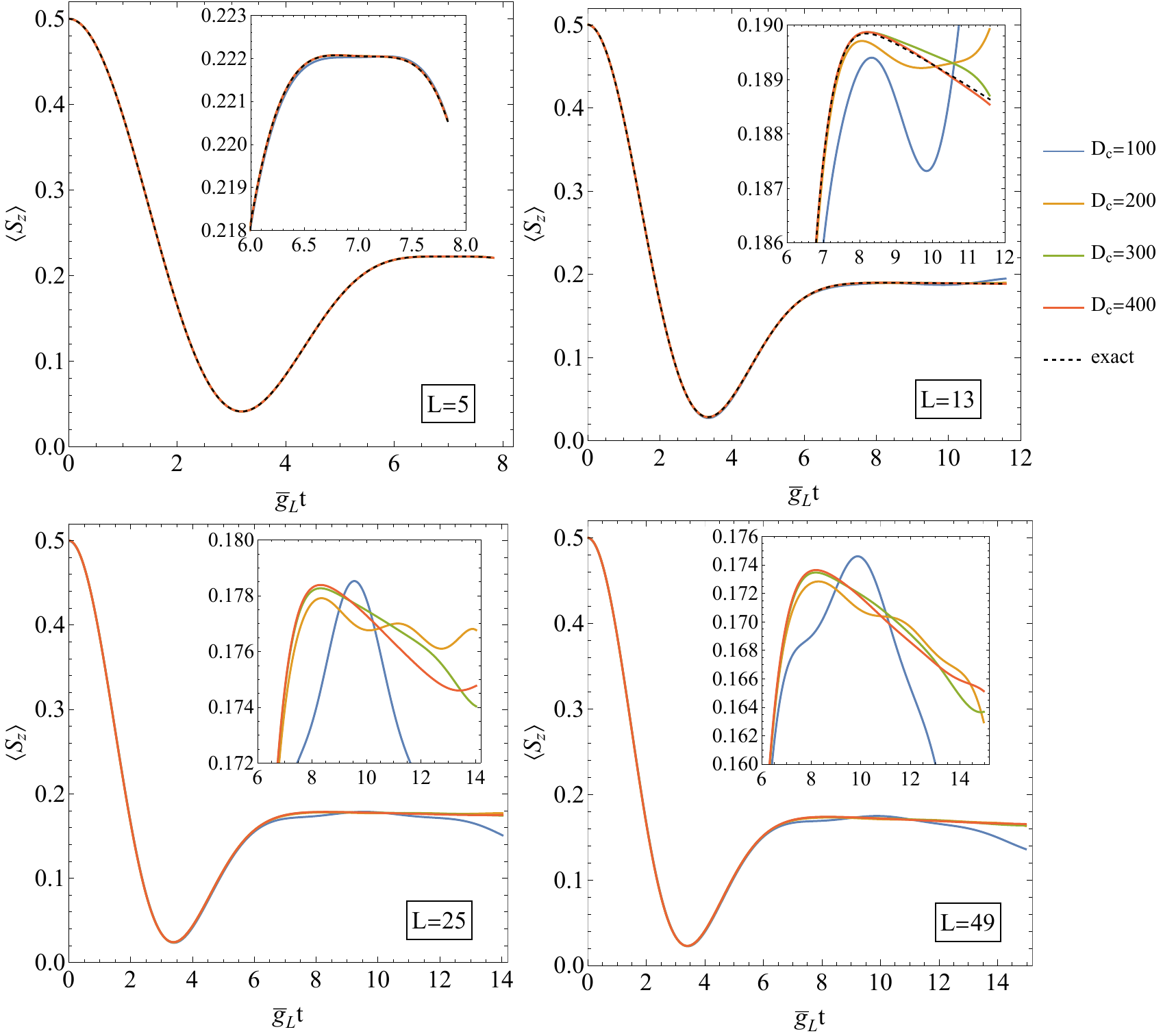}
     \caption{Central spin polarization computed by the EDM method with different bond dimension cutoffs $D_c$ for various numbers of bath spins included. The exact numerical solution is also shown for small numbers of bath spins ($L\le 13$). The insets zoom in the plateau regions of the curves. The parameters used are the same as in Fig.~\ref{fig:Gaudin}.}
     \label{fig:GM-comp-app}
 \end{figure}

 In the calculation, we introduce a cutoff of the bond dimension $D_{c}$ to manage the computation cost, which would induce error in addition to the truncation precision $\xi$ in the SVD and the expansion error $\mathcal{O}(\epsilon^3)$ in Eq.~\eqref{eq:high-order-app}. Figure~\ref{fig:GM-comp-app} compares the results obtained with various bond dimension cutoff, with the exact numerical results also shown for small numbers of bath spins ($L\le 13$). One can find that when $L=5$, different cutoffs yield almost the same decoherence dynamics, all in good agreement with the exact solution, which indicates that the non-Markovian dynamics in this case can be described by a small number of EDMs. However, for $L=13$ bath spins considered, the deviation from the exact solution is noticeable when the bond dimension cutoff is small $D_c\le 300$. The deviation decreases with increasing the cutoff $D_c$. For larger spin baths ($L\ge 25$), no exact solution is available. However, the convergent dynamics indicates that the EDM calculation with bond dimension cutoff $D_c=400$ (adopted in the main text) is reliable up to errors $\sim 10^{-3}$.   

\begin{figure}
    \centering
    \includegraphics[width=0.9\linewidth]{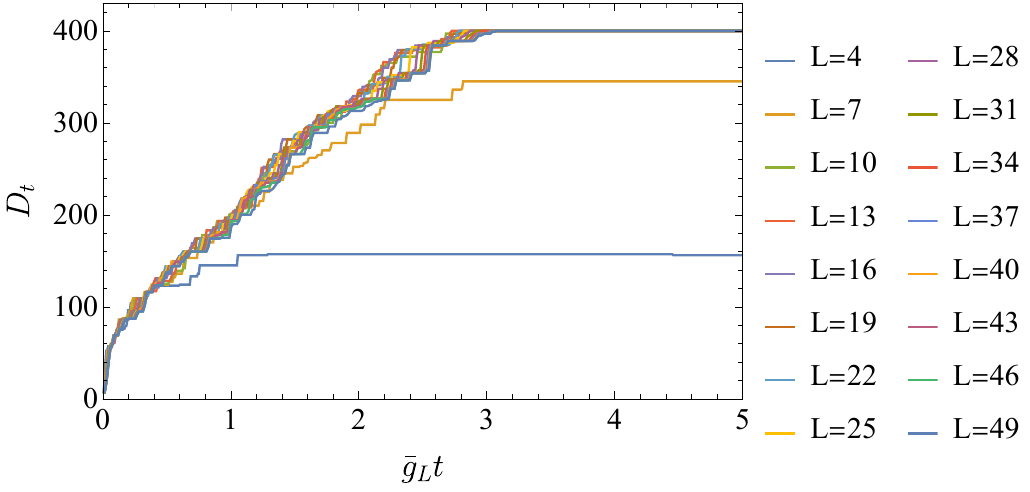}
    \caption{The bond dimension in EDMs $\rho^{\bm{\Phi}^{(T)}}_{L}$ as a function of the rescaled time $\bar{g}_L t$ for various numbers of bath spins included. The parameters are the same as in Fig.~\ref{fig:Gaudin}.}
    \label{fig:collapse-app}
\end{figure}

\subsection{Universal behavior of bond dimension increase} \label{Append_universal}
In the main text, we have seen that the bond dimension $D_t$ in different EDM tensors $\rho^{\bm{\Phi}^{(T)}}_{L}$ as functions of the rescaled evolution time  $\bar{g}_L t$ are nearly identical for different $L$.  Figure~\ref{fig:collapse-app} shows more data. The increase of the bond dimension collapses to a universal curve until it meets the cutoff (set as $D_c=400$) or saturates for small baths ($L\le 7$). This phenomenon demonstrates that with a given coupling strength between the system and the bath, the growth of bond dimension is independent of the bath size.  As the number of bath spins increases, the degrees of freedom associated with the bath spins increase exponentially. However, the growth of the bond dimension in EDMs is nearly linear in time, which confirms the linearly increasing bond dimension theorem.

\bibliography{TensorNetworkOPS}
\end{document}